\documentclass{llncs}
\newif\iflong
\newif\ifshort

\longtrue

\iflong
\else
\shorttrue
\fi

\let\oldbibliography\bibliography
\renewcommand{\bibliography}[1]{{%
  \let\chapter\section
  \oldbibliography{#1}}}

\usepackage[percent]{overpic}
\usepackage{amssymb,amsfonts,epsf}
\usepackage{epsfig}
\usepackage{graphics}
\usepackage{epsfig}
\usepackage{alltt}
\usepackage{color}
\usepackage[utf8]{inputenc}
\usepackage[T1]{fontenc}
\usepackage{lmodern}
\usepackage[numbers,sort&compress]{natbib} 
\usepackage{mathtools}
\usepackage{amsmath,amssymb}
\usepackage{enumerate}
\usepackage{graphicx}
\usepackage{booktabs}
\usepackage{xcolor}
\usepackage[textsize=footnotesize]{todonotes}
\usepackage{verbatim}
\usepackage{paralist}
\usepackage{multicol}
\usepackage{amsmath,comment}

\newtheorem{observation}{Observation}
\newtheorem{assumption}{Assumption}
\newcommand{\paramproblem}[4]{\noindent {\sc #1}
\\
{\bf Given:} #2\\
{\bf Parameter:} #3\\
{\bf Question:} #4}

\newcommand{\pc}{$P$-component}
\newcommand{\pcs}{$P$-components}
\newcommand{\btp}{{\sc Balanced Tree Partitioning}}
\newcommand{\threebtp}{{\sc Balanced Degree-3 Tree Partitioning}}
\newcommand{\threegtp}{{\sc Degree-3 Tree Partitioning}}
\newcommand{\gtp}{{\sc Tree Partitioning}}
\newcommand{\mcc}{{\sc $k$-Multi-Colored Clique}}

\newcommand{\Oh}{{\mathcal O}}
\newcommand{\nat}{\mathbb{N}}
\newcommand{\integ}{\mathbb{Z}^+}
\newcommand{\pp}{\Diamond}
\newcommand{\CR}{{\bf Check-Realizability}}
\newcommand{\srep}{$\sigma$-representation}
\newcommand{\sreps}{$\sigma$-representations}

\newcommand{\true}{{\sc true}}
\newcommand{\false}{{\sc false}}

\newcommand{\Pol}{\mbox{$\mathcal P$}}
\newcommand{\APX}{\mbox{$\mathcal APX$}}
\newcommand{\fpt}{\mbox{$\mathcal{FPT}$}}
\newcommand{\NP}{\mbox{$\mathcal{NP}$}}
\newcommand{\f}{\mathcal{\lambda}_P}

\def\eg{{\em e.g.}}

\def\ie{{\em i.e.}}
\def\etal{{\em et al.}}

\begin{document}
\pagestyle{plain}

\title{The Complexity of Tree Partitioning}
\author{{\sc Zhao An}\inst{1} \ \ {\sc Qilong Feng}\inst{1} \ \ {\sc Iyad Kanj}\inst{2} \ \ {\sc Ge
Xia}\inst{3}}

\institute{
School of Information Science and Engineering,
Central South University, China.
{\tt
anzhao1990@126.com, csufeng@csu.edu.cn}
\and
School of Computing, DePaul
University, Chicago, IL. {\tt
ikanj@cs.depaul.edu}
\and
Dept. of Computer
Science, Lafayette College, Easton, PA. {\tt xiag@lafayette.edu}}

\date{}

\maketitle
\setcounter{footnote}{0}
\begin{abstract}
Given a tree $T$ on $n$ vertices, and $k, b, s_1, \ldots, s_b  \in \nat$, the \gtp~problem asks if at most $k$ edges can be removed from $T$ so that the resulting components can be grouped into $b$ groups such that the number of vertices in group $i$ is $s_i$, for $i =1, \ldots, b$. The case when $s_1=\cdots =s_b =n/b$, referred to as the \btp~problem, was shown to be \NP-complete for trees of maximum degree at most 5, and the complexity of the problem for trees of maximum degree 4 and 3 was posed as an open question.~The parameterized complexity of \btp~was also posed as an open question in another work.

In this paper, we answer both open questions negatively. We show that \btp~(and hence, \gtp) is \NP-complete for trees of maximum degree 3, thus closing the door on the complexity of \btp, as the simple case when $T$ is a path is in \Pol. In terms of the parameterized complexity of the problems, we show that both \btp~and \gtp~are $W[1]$-complete. Finally, using a compact representation of the solution space for an instance of the problem, we present a dynamic programming algorithm for \gtp~(and hence, for \btp) that runs in subexponential-time $2^{O(\sqrt{n})}$, adding a natural problem to the list of problems that can be solved in subexponential time.
\end{abstract}

\section{Introduction} \label{sec:intro}
\paragraph{{\bf Problem Definition and Motivation.}} We consider the \gtp~problem defined as follows: \\

\paramproblem{\gtp} {A tree $T$; $k, b, s_1, \ldots, s_b  \in \nat$}{$k$}{Does there exist a subset $E' \subseteq E(T)$ of at most $k$ edges such that the components of $T-E'$ can be grouped into $b$ groups, where group $i$ contains $s_i$ vertices, for $i=1, \ldots, b$?}

The special case of the problem when $s_1=\cdots =s_b =|V(T)|/b$~\iflong~(\ie, when all the groups have the same size)~\fi is referred to as \btp\iflong\footnote{In a variant of the \btp~problem, the group sizes in the solution are allowed to differ by 1. All the results in this paper still hold for this variant of the problem.}\fi.

The two problems are special cases of the {\sc Balanced Graph Partitioning} problem, which has applications in the areas of parallel computing~\cite{application1}, computer vision~\cite{application1}, VLSI circuit design~\cite{application2}, route planning~\cite{application3}, and image processing\ifshort~\cite{application4}\fi \iflong~\cite{application4,application5}\fi. \iflong In such applications, the problem that arises is to schedule $n$ objects (modeled as graph vertices) into $b$ groups such that the intercommunication between the objects (modeled as graph edges) is minimized.~\fi The special case of {\sc Balanced Graph Partitioning}, corresponding to $b=2$, is the well-known~\NP-complete problem {\sc Bisection}~\cite{gj}.
The {\sc Balanced Graph Partitioning} problem has received a lot of attention from the area of approximation theory (for instance, see~\cite{feldmannthesis,racke,racke1}). Moreover, the complexity and the approximability of the problem restricted to special graph classes, such as grids, trees, and bounded degree trees~\cite{feldmannthesis,feldmanngrid,feldmann,macgregor}, have been studied.

\paragraph{{\bf Our Results.}} We study the complexity and the parameterized complexity of \gtp~and \btp, and design subexponential time algorithms for these problems. Our results are:

\begin{itemize}
\item[(A)] We prove that \btp, and hence \gtp, is \NP-complete for trees with maximum degree at most 3. This answers an open question in~\cite{feldmann} about the complexity of \btp~for trees of maximum degree 4 and 3, after they had shown the \NP-completeness of the problem for trees of maximum degree at most 5. This also closes the door on the complexity of these problems on trees, as the simple case when the tree is a path is in \Pol.

\item[(B)] We prove that both \gtp~and \btp~are $W[1]$-complete. This answers an open question in~\cite{bevern}. \iflong We observe that, for trees, the removal of $k$ edges results in $k+1$ components. Since the number of groups $b$ is at most $k+1$ (an upper bound on the number of components), the $W[1]$-hardness results with respect to parameter $k$ imply the $W[1]$-hardness of the problems with respect to the parameter-combination $(b, k)$.\fi  We also prove the membership of the problems in the class $W[1]$, using the characterization of $W[1]$ given by Chen \etal~\cite{yijia}.

\item[(C)] We present an exact subexponential-time algorithm for \gtp, and hence for \btp, that runs in time $2^{\Oh(\sqrt{n})}$, where $n$ is the number of vertices in the tree.
\end{itemize}

For the lack of space, many details and proofs in this paper have been omitted.

\paragraph{{\bf Related Work and Our Contributions.}}
Feldmann and Foschini~\cite{feldmann} studied \btp. They showed that the problem is \NP-complete for trees of maximum degree at most 5, and left the question about the complexity of the problem for maximum degree 4 and 3 open. Whereas the reduction used in the current paper to prove the \NP-hardness of \btp~on trees of maximum degree at most 3 starts from the same problem ({\sc 3-Partition}) as in~\cite{feldmann}, and is inspired by their construction, the reduction in this paper is much more involved in terms of the gadgets employed and the correctness proofs. \iflong Feldmann and Foschini~\cite{feldmann} also showed that \btp~is \APX-hard, with respect to the size of the cut (\ie, the number of edges cut), for trees of maximum degree at most 7.  They also considered the problem where one seeks a partitioning that approximates a balanced partitioning to within factor $(1+\epsilon)$ (\ie, the size of each group is within a factor of $(1+\epsilon)$ from the size $n/b$ of a balanced partitioning), and showed that this problem admits a PTAS with respect to this notion of approximation. MacGregor, in his Ph.D. thesis~\cite{macgregor} in 1978, proposed a greedy algorithm that approximates \btp~on trees of constant maximum degree to within factor $\Oh(\lg{n}/b)$ from the optimal solution. Refer to the Ph.D. thesis of Feldmann~\cite{feldmannthesis} for more information on polynomial-time algorithms and approximation algorithms for \btp~on special graph classes. \fi

Bevern \etal~\cite{bevern} showed that the parameterized complexity of {\sc Balanced Graph Partitioning} is $W[1]$-hard when parameterized by the combined parameters $(k, \mu)$, where $k$ is (an upper bound on) the cut size, and $\mu$ is (an upper bound on) the number of resulting components after the cut. It was observed in~\cite{bevern}, however, that the employed \fpt-reduction yields graphs of unbounded treewidth, which motivated the authors to ask about the parameterized complexity of the problem for graphs of bounded treewidth, and in particular for trees. We answer their question by showing that the problem is $W[1]$-complete.

Bevern \etal~\cite{bevern} also showed that {\sc Balanced Graph Partitioning} is $W[1]$-hard on forests by a reduction from the {\sc Unary Bin Packing} problem, which was shown to be $W[1]$-hard in~\cite{binpacking}. We note that the disconnectedness of the forest is crucial to their reduction, as they  represent each number $x$ in an instance of {\sc Bin Packing} as a separate path of $x$ vertices.
For \btp, in contrast to {\sc Unary Bin Packing} (and hence, to {\sc Balanced Graph Partitioning} on forests), the difficulty is not in grouping the components into groups (bins) because enumerating all possible distributions of $k+1$ components (resulting from cutting $k$ edges) into $b \leq k+1$ groups can be done in \fpt-time; the difficulty, however, stems from not knowing which tree edges to cut. The \fpt-reduction we use to show the $W[1]$-hardness is substantially different from both of those in~\cite{bevern,binpacking}, even though we use the idea of non-averaging sets in our construction---a well-studied notion in the literature (e.g., see~\cite{nonaveraging}), which was used for the $W[1]$-hardness result of {\sc Unary Bin Packing} in~\cite{binpacking}.

Many results in the literature have shown that certain \NP-hard graph problems are solvable in subexponential time. Some of these rely on topological properties of the underlying graph that guarantee the existence of a balanced graph-separator of sub-linear size, which can then be exploited in a divide-and-conquer approach (\eg, see~\cite{Demaine,genusicalp}). There are certain problems on restricted graph classes that resist such approaches due to the the problem specifications; designing subexponential-time algorithms for such problems usually require exploiting certain properties of the solution itself, in addition to properties of the graph class (see~\cite{fominsocg,marxklein} for such recent results). In the case of \gtp~and \btp, since every tree has a balanced separator consisting of a single vertex, yet the two problems remain \NP-hard on trees, clearly a divide-and-conquer approach based solely on balanced separators does not yield subexponential-time algorithms for these problems. To design subexponential-time  algorithms for them, we rely on the observation that the number of possible partitions of an integer $n \in \nat$ is subexponential in $n$; this allows for a ``compact representation'' of all solutions using a solution space of size $2^{\Oh(\sqrt{n})}$, enabling a dynamic programming approach that solves the problems within the same time upper bound.

\ifshort
\paragraph{{\bf Terminologies.}}
 We refer the reader to~\cite{fptbook,diestel} for more information about graph theory and parameterized complexity.

Let $T$ be a rooted tree.  For an edge $e=uv$ in $T$ such that $u$ is the parent of $v$, by the subtree of $T$ {\em below} $e$ we mean the subtree $T_v$ of $T$ rooted at $v$. For two edges $e, e'$ in $T$, $e$ is said to be {\em below} $e'$ if $e$ in an edge of the subtree of $T$ below $e'$.  \iflong A {\em binary tree} is a rooted tree in which each vertex has at most two children.\fi A {\em nice binary tree} $T$ is defined recursively as follows. If $|V(T)| \leq 1$ then $T$ is a nice binary tree. If $V(T) > 1$, then $T$ is nice if (1) each of the left-subtree and right-subtree of $T$ is nice and (2) the sizes of the left-subtree and the right-subtree differ by at most 1. For any $n \in \nat$, there is a nice binary tree of order $n$. A {\em star} $S$ is a tree consisting of a single vertex $r$, referred to as the root of the star, attached to degree-1 vertices, referred to each as a {\em star-leaf}; we refer to an edge between $r$ and a leaf in $S$ as a {\em star-edge}; we refer to a subtree of $S$ containing $r$ as a {\em substar} of $S$.

A {\em solution} $P$ to an instance $(T, k, b, s_1, \ldots, s_b)$ of \gtp~is a pair $(E_P, \mathcal{\lambda}_P)$, where $E_P$ is a set of $k$ edges in $T$, and $\mathcal{\lambda}_P$ is an assignment that maps the connected components in $T-E_P$ into $b$ groups so that the total number of vertices assigned to group $i$ is $s_i$, for $i \in [b]$. We call a connected component in $T-E_P$ a {\em \pc{}}, and denote by $C_P$ the set of all \pcs{}.

By a {\em cut in a tree} $T$ we mean the removal of an edge from $T$. A solution $P =(E_P, \mathcal{\lambda}_P)$ to $(T, k, b, s_1, \ldots, s_b)$ \emph{cuts} an edge $e$ in $T$ if $e \in E_P$. Let $T'$ be a subtree of $T$ such that $P$ cuts at least one edge in $T'$. By a \emph{lowest} \pc{} in $T'$ we mean a subtree $T''$ below an edge $e$ of $T'$ such that $T''$ is a \pc{} (\ie, $P$ does not cut any edge below $e$ in $T'$).

The restriction of \gtp~to instances in which $s_1=\cdots=s_b=|T|/b$ is denoted \btp; an instance of \btp~is a triplet $(T, k, b)$. The restriction of \gtp~and \btp~to trees of maximum degree at most 3 are denoted \threegtp~and \threebtp, respectively.
For $\ell \in \nat$, we write $[\ell]$ for the set $\{1, \ldots, \ell\}$.
\fi

\iflong
\section{Preliminaries}\label{sec:prelim}

\paragraph{{\bf Graphs, Trees and Stars.}} A \emph{tree} $T$ is an undirected acyclic graph. A \emph{forest} is a disjoint union of trees. We write $V(T)$ and $E(T)$ for the vertex-set and edge-set of $T$, respectively. By $|T|$ we denote the \emph{order} of $T$, which is $|V(T)|$. A \emph{subtree} of $T$ is a tree induced by a subset of $V(T)$.  For a set of edges $E'$ in $T$, by $T-E'$ we denote the forest whose vertex-set is $V(T)$ and edge-set is $E(T)\setminus E'$. For two forests $F$ and $F'$, we write $F-F'$ for the forest induced by the vertex-set $V(F)\setminus V(F')$.

A \emph{rooted} tree is a tree with a vertex designated as the root. For a rooted tree $T$, we can define the parent-child and ancestor-descendant relations on the vertex-set of $T$ in a natural way. For a rooted tree $T$ and a vertex $v \in V(T)$, we write $T_v$ for the subtree of $T$ rooted at $v$.

A {\em binary tree} is a rooted tree in which each vertex has at most two children. A {\em nice binary tree} $T$ is a binary tree defined recursively as follows. If $|V(T)| \leq 1$ then $T$ is a nice binary tree. If $V(T) > 1$, then $T$ is nice if (1) each of the left-subtree and right-subtree of $T$ is nice and (2) the sizes of the left-subtree and the right-subtree differ by at most 1. It is clear that for any $n \in \nat$, there is a nice binary tree of order $n$.

Let $T$ be a rooted tree.  For an edge $e=uv$ in $T$ such that $u$ is the parent of $v$, by the subtree of $T$ {\em below} $e$ we mean the subtree $T_v$ of $T$ rooted at $v$. For two edges $e, e'$ in $T$, $e$ is said to be {\em below} $e'$ if $e$ in an edge of the subtree of $T$ below $e'$.

A {\em star} $S$ is a tree consisting of a single vertex $r$, referred to as the root of the star, attached to degree-1 vertices, referred to each as a {\em star-leaf}; we refer to an edge between $r$ and a leaf in $S$ as a {\em star-edge}; we refer to a subtree of $S$ containing $r$ as a {\em substar} of $S$.

\paragraph{{\bf \gtp~and Its Related Terminologies.}}


A {\em solution} $P$ to an instance $(T, k, b, s_1, \ldots, s_b)$ of \gtp~is a pair $(E_P, \mathcal{\lambda}_P)$, where $E_P$ is a set of $k$ edges in $T$, and $\mathcal{\lambda}_P$ is an assignment that maps the connected components in $T-E_P$ into $b$ groups so that the total number of vertices assigned to group $i$ is $s_i$, for $i \in [b]$. We call a connected component in $T-E_P$ a {\em \pc{}}, and denote by $C_P$ the set of all \pcs{} in $T-E_P$.

By a {\em cut in a tree} $T$ we mean the removal of an edge from $T$. We say that a solution $P =(E_P, \mathcal{\lambda}_P)$ to an instance  $(T, k, b, s_1, \ldots, s_b)$ of \gtp~\emph{cuts} an edge $e$ in $T$ if $e \in E_P$. For a subtree $T'$ of $T$ such that $P$ cuts at least one edge in $T'$, by a \emph{lowest} \pc{} in $T'$ we mean a subtree $T''$ below an edge $e$ of $T'$ such that $T''$ is a \pc{} (\ie, $P$ does not cut any edge below $e$ in $T'$).

The restriction of \gtp~to instances in which $s_1=\cdots=s_b=|T|/b$ is denoted \btp; an instance of \btp~is specified as a triplet $(T, k, b)$. The restriction of \gtp~and \btp~to trees of maximum degree at most 3 are denoted \threegtp~and \threebtp, respectively.

\paragraph{{\bf Parameterized Complexity.}} A {\it parameterized problem} is a set of instances of the form $(x, k)$, where $x \in \Sigma^*$ for a finite alphabet set $\Sigma$, and $k \in \nat$ is the {\em parameter}.
A parameterized problem $Q$ is {\it fixed parameter tractable} (\fpt), if there exists an algorithm that on input $(x, k)$
decides if $(x, k)$ is a yes-instance of $Q$ in time $f(k)|x|^{\Oh(1)}$,
where $f$ is a computable function; we will denote by {\em \fpt-time} a running time of the form $f(k)|x|^{\Oh(1)}$. A parameterized problem $Q$
is {\it \fpt-reducible} to a parameterized problem $Q'$, written $Q \preceq_{fpt} Q'$, if there is an algorithm that transforms each instance $(x, k)$ of $Q$
into an instance $(x', g(k))$ of
$Q'$ in \fpt-time, where $g$ is a computable function, and such that $(x, k) \in Q$ if and
only if $(x', g(k)) \in Q'$. A parameterized complexity hierarchy, {\it the $W$-hierarchy} $\bigcup_{t
\geq 0} W[t]$, was introduced based on the notion of \fpt-reduction, in which the $0$-th level $W[0]$ is the class $\fpt$.  It is commonly believed that $W[1] \neq \fpt$. For more information about parameterized complexity, we refer the reader to~\cite{fptbook,grohebook,rolfbook}.

For $\ell \in \nat$, we write $[\ell]$ for the set $\{1, \ldots, \ell\}$.
\fi
 \iflong
\section{\threebtp~and \threegtp~are~\NP-complete}\label{sec:npcomplete}
\fi
 \ifshort
\section{\NP-completeness}\label{sec:npcomplete}
\fi
In this section, we show that \threebtp{}, and hence \threegtp, is~\NP-complete. Without loss of generality, we will consider the version of \threebtp{} in which we ask for a cut of size exactly $k$, as opposed to at most $k$; it is easy to see that the two problems are polynomial-time reducible to one another.

To prove that \threebtp{} is \NP-hard, we will show that the strong \NP-hard problem {\sc 3-Partition}~\cite{gj} is polynomial-time reducible to it. Our reduction is inspired by the construction of Feldmann and Foschini~\cite{feldmann}. Whereas the construction in~\cite{feldmann} uses gadgets each consisting of five chains joined at a vertex, the construction in this paper uses gadgets consisting of nearly-complete binary trees, that we refer to as nice binary trees. The idea behind using nice binary trees is that we can combine them to construct a degree-3 tree in which the cuts must happen at specific edges in order to produce components of certain sizes.

An instance of the {\sc 3-Partition} problem consists of an integer $s > 0$ and a collection $S=\langle a_1, \ldots, a_{3k} \rangle$ of $3k$ positive integers, where each $a_i$ satisfies $s/4 < a_i < s/2$, for $i\in [3k]$. The problem is to decide whether $S$ can be partitioned into $k$ groups $S_1, \ldots, S_k$, each of cardinality 3, such that the sum of the elements in each $S_i$ is $s$, for $i \in [k]$.  

\iflong
Let $(S=\langle a_1, \ldots, a_{3k} \rangle, s)$ be an instance of {\sc 3-Partition}. If we multiply $s$ and each $a_i \in S$, $i \in[3k]$, by any fixed $x \in \integ$, we obtain an equivalent instance of {\sc 3-Partition}. For the purpose of this reduction, we will apply the following (polynomial-time) transformation that either rejects the instance, or transforms it into an equivalent instance of {\sc 3-Partition}:
\begin{enumerate}
\item Multiply $s$ and each element in $S$ by $4$. As a consequence, for each element $a_i$ in $S$, we now have $s/4 + 1 \leq a_i \leq s/2 -1$ because $s/4 < a_i < s/2$ (by the problem definition) and $s$ is divisible by $4$.

\item If there is an element $a_i = s/2-1$, then reject the instance because $a_i$ cannot be grouped with two other elements in $S$ to make a group of size $s$ (because every element in $S$ is at least $s/4+1$). If there is an element $a_i = s/2-2$, then the only way that $a_i$ can be grouped with two elements in $S$ to give a total sum of $s$, is to group $a_i$ with two elements each of size $s/4+1$. If there are two other elements each of value $s/4+1$, then remove $a_i$ and these two elements; otherwise reject the instance. Now we have $s > 4$ and $s/4+1 \leq a_i \leq s/2-3$, for $i \in [3k]$.

\item Multiply $s$ and each element in $S$ by $6k$. Now we have $s > 24k$ and $s/4+6k \leq a_i \leq s/2-18k$, for $i \in [3k]$.
\end{enumerate}

For later use, we summarize all the above in an assumption below:

\begin{assumption}\label{assumption:ass}
We assume that: (1) $s$ is a multiple of 4; (2) $s > 24k$; and (3) $s/4+6k \leq a_i \leq s/2-18k$, for $i \in [3k]$.
\end{assumption}
\fi

For the reduction, we construct a degree-3 tree $T$ as follows. For each $a_i \in S$, we create a binary tree $T_i$, whose left subtree $L_i$ is a nice binary tree of size $a_i$, and whose right subtree $R_i$ is a nice binary tree of size $s-2$. We denote by $R^l_i$ and $R^r_i$ the left and right subtrees of $R_i$, respectively. Let $H=(p_1, \ldots, p_{3k})$ be a path on $3k$ vertices. The tree $T$ is constructed by adding an edge between each $p_i$ in $H$ and the root of $T_i$, for $i\in [3k]$. See Figure~\ref{fig:tree} for illustration. It is clear from the construction that $T$ is a degree-3 tree of $4k\cdot s$ vertices, since each $T_i$ has size $a_i + s-1$ and $P$ has $3k$ vertices. We will show that $(S, s)$ is a yes-instance of {\sc 3-Partition} if and only if the instance $I=(T, 6k-1, b=4k)$ is a yes-instance of \threebtp. \iflong We will prove the aforementioned statement by proving a sequence of lemmas.\fi


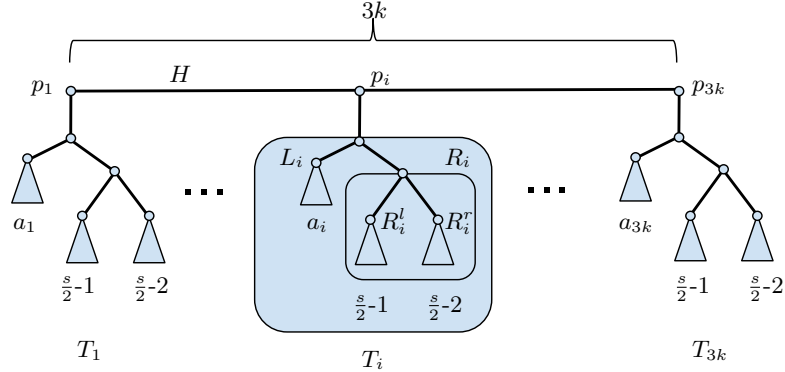
\begin{figure}[htbp]
\begin{overpic}[clip, trim=1cm 6.8cm 1cm 1.5cm, width=1.00\textwidth,tics=10]{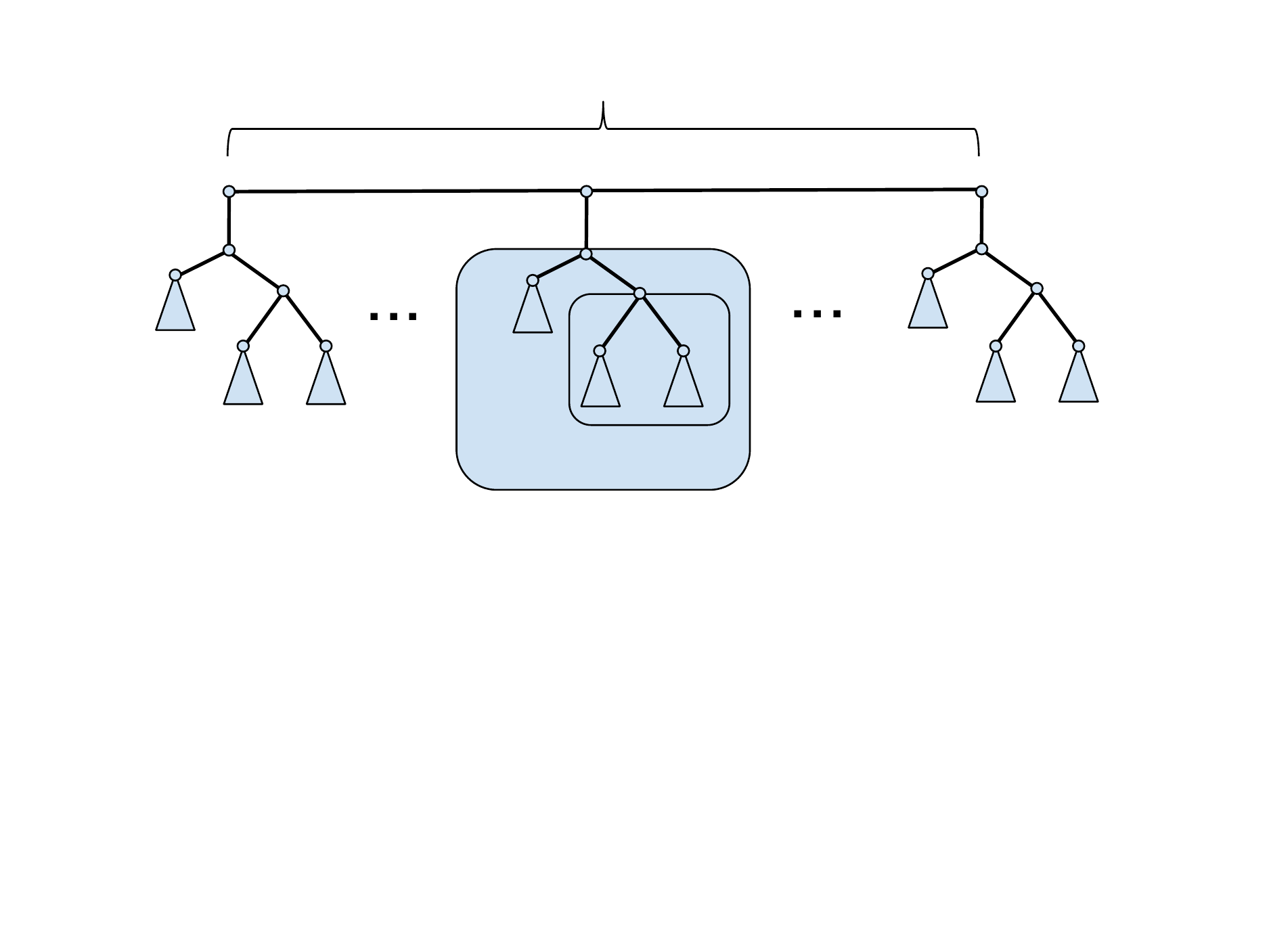}
 \put (25,28) {$H$}
 \put (46,35) {$3k$}
 \put (10,27) {$p_1$}
 \put (82,27) {$p_{3k}$}
 \put (47,28) {$p_{i}$}

 \put (15,-2) {$T_{1}$}
 \put (46,-3) {$T_{i}$}
 \put (82,-2) {$T_{3k}$}

 \put (8,12) {$a_{1}$}
 \put (40,12) {$a_{i}$}
 \put (74,12) {$a_{3k}$}

 \put (13,5) {$\frac{s}{2}$-1}
 \put (21,5) {$\frac{s}{2}$-2}

 \put (45,3) {$\frac{s}{2}$-1}
 \put (53,3) {$\frac{s}{2}$-2}

 \put (80,5) {$\frac{s}{2}$-1}
 \put (88,5) {$\frac{s}{2}$-2}

 \put (55,19) {$R_i$}
 \put (37,19) {$L_i$}
 \put (48,12) {$R^l_i$}
 \put (55,12) {$R^r_i$}
\end{overpic}
\caption{Illustration of the construction of the tree $T$.}
\label{fig:tree}
\end{figure}

Note that the size of $T$ is $4k\cdot s$, and hence, if the vertices in $T$ can be grouped into $4k$ groups of equal size, then each group must contain $s$ vertices.
From the aforementioned statement, it follows that at least one cut is required in each tree $T_i$ because the size of each $T_i$ is $a_i + s-1 > s$.

\ifshort
Suppose that the instance $I$ has a solution $P$ that cuts $6k-1$ edges in $T$.
\fi

\iflong
Suppose that the instance $I$ has a solution $P$ that cuts $6k-1$ edges in $T$.   Let $C'_P \subseteq C_P$ be such that each \pc{} in $C'_P$ is contained in some $T_i$. We call the forest resulting after all \pcs{} in $C'_P$ are removed from $T$ a {\em partial-$T$}, and the remaining portion of each $T_i$ (respectively $R_i$) in a partial-$T$ a {\em partial}-$T_i$ (respectively {\em partial}-$R_i$).

\begin{lemma} \label{lem:1}
Let $m \in \nat$, and let $T'$ be a partial-$T$. Then $T'$ cannot satisfy all of the following conditions:
\begin{enumerate}

    \item $|\mathcal{T}| = m$, where $\mathcal{T}$ is the set of partial-$T_i$'s in $T'$, each of size at least $s-1+a_i - 6k$ and its partial-$R_i$ is not a \pc{}.
    \item $T'$ contains $2m$ \pcs{}.
    \item $T'$ contains a partial-$T_{i_0} \notin \mathcal{T}$, $i_0 \in [3k]$, of size at least $a_{i_0}+1-6k$.
\end{enumerate}
\end{lemma}

\begin{proof}
Suppose that $T'$ satisfies all the above conditions. We will show that at least one \pc{} has size greater than $s$.

First observe that each partial-$T_i$ in $\mathcal{T}$ has at least one \pc{} because its size is at least $s-1+a_i - 6k > s$. Let $\mathcal{T}_1$ be the set of  partial-$T_i$'s in $\mathcal{T}$ each containing exactly one \pc{}. Let $\mathcal{T}_2$ be set of partial-$T_i$ in $\mathcal{T}$ each containing two or more \pcs{}. Remove all the \pcs{} contained in $\mathcal{T}$ from $T'$. At least $|\mathcal{T}_1|+2|\mathcal{T}_2|$ \pcs{} are removed from $T'$, and hence at most $2m-|\mathcal{T}_1|-2|\mathcal{T}_2| = |\mathcal{T}_1|$ \pcs{} remain in $T'$, because $m=|\mathcal{T}_1|+|\mathcal{T}_2|$.

For each partial-$T_i$ in $\mathcal{T}_1$, its partial-$R_i$ cannot be removed as a \pc{} because, by condition (1) its partial-$R_i$ is not a \pc{}. Therefore, the maximum possible size of the \pc{} removed from $T_i$ is $s/2-1$ (the left subtree of $R_i$). This means that the size of the remaining portion of each partial-$T_i$ in $\mathcal{T}_1$ is at least $s-1+a_i - 6k - (s/2-1) = s/2+a_i-6k$.

Therefore, $T'$ still has $|\mathcal{T}_1|$-many partial-$T_i$'s, each of size at last $s/2+a_i-6k$, and no \pc{} is contained in any of these partial-$T_i$'s; otherwise they would belong to $\mathcal{T}_2$. Since there are $|\mathcal{T}_1|$-many \pcs{} left in $T'$, each \pc{} must contain exactly one partial-$T_i$ from $\mathcal{T}_1$ because leaving any two of them connected would result in a \pc{} of size more than $s$. One of these $|\mathcal{T}_1|$-many \pcs{} must also contain a partial-$T_{i_0}$ of size $a_{i_0}+1$ by condition (3). This \pc{}'s size will be at least $s/2+a_i-6k + (a_{i_0} + 1-6k) > s$ because, by Assumption~\ref{assumption:ass}, $a_i \geq s/4+6k$, for $i \in [3k]$.
\qed
\end{proof}

\begin{lemma} \label{lem:2}
For $i \in [3k]$, $R_i$ is not a \pc{} in $C_P$.
\end{lemma}

\begin{proof}
Suppose that $C_P$ contains $h$-many $R_i$'s as \pcs{}, where $h\geq 1$.  Since $|R_i| = s-2$, for each such $R_i$, there must exist at least one other \pc{} $X_i$ in $C_P$, where $|X_i| \leq 2$, that is assigned to the same group as $R_i$; call such an $X_i$ a {\em tiny \pc{}}. Let $T'$ be the partial-$T$ resulting from $T$ after all such $R_i$'s and $X_i$'s have been removed from $T$. We will show that $T'$ satisfies all three conditions of Lemma~\ref{lem:1}, which implies that at least one \pc{} in $C_P$ would have size greater than $s$, thus deriving a contradiction.

First, for each $T_i$ whose $R_i$ is not removed above, at most $h$-many tiny \pcs{} could be removed from $T_i$. Since $h \leq 3k$ and each tiny \pc{} has size at most 2, the resulting partial-$T_i$ in $T'$ has size at least $s-1+a_i - 6k$. Therefore, there are $(3k-h)$-many partial-$T_i$'s satisfying condition (1) of Lemma~\ref{lem:1}. Second, $T'$ contains $6k-2h$ \pcs{} satisfying condition (2) of Lemma~\ref{lem:1}, because $h$-many $R_i$'s and $h$-many $X_i$'s are removed from $T$ to produce $T'$. Finally, for each $T_i$ whose $R_i$ is removed above, the resulting partial-$T_i$ in $T'$ has size at least $a_{i}+1-6k$, because in addition to $R_i$, at most $h$-many tiny \pcs{}, each of size at most 2, could be removed from $T_i$ and $h \leq 3k$. Since $h \geq 1$, $T'$ satisfies condition (3) of Lemma~\ref{lem:1}. \qed
\end{proof}

\begin{lemma}\label{lem:3}
For $i \in [3k]$, $T_i$ does not contain a lowest \pc{} of size less than $s/4$.
\end{lemma}

\begin{proof}
Suppose that a $T_i$ contains a lowest \pc{} $C_1$ of size less than $s/4$. Remove $C_1$ from $T_i$. The resulting partial-$T_i$ has size at least $s-1+a_i-(s/4-1) = 3s/4+a_i > s$. Therefore, there is at least one more \pc{} contained in the partial-$T_i$. Remove another lowest \pc{} $C_2$ from this partial-$T_i$. If $C_1$ is in $L_i$, then sine $C_2 \neq R_i$ by Lemma~\ref{lem:2}, the largest \pc{} $C_2$ can be is the left subtree of $R_i$ of size $s/2-1$. If $C_1$ is contained in $R_i$, then the largest \pc{} $C_2$ can be is the remaining portion of $R_i$, which has size $s-2-|C_1|$. In either case, the partial-$T_i$ resulting after removing $C_1$ and $C_2$ has size at least $a_i+1$. Let $T'$ be the partial-$T$ after $C_1$ and $C_2$ are removed.

It is clear that (1) $T'$ has $(3k-1)$-many $T_i$'s of size at least $s-1+a_i$ and by Lemma~\ref{lem:2} the $R_i$ in each of these $T_i$'s is not a \pc{}; (2) $T'$ contains $6k-2$ \pcs{} (because $2$ \pcs{} are removed from $T$ to produce $T'$), and (3) $T'$ has a partial-$T_{i_0}$ of size at least $a_{i_0}+1$. By Lemma~\ref{lem:1}, at least one component in $C_P$ has size greater than $s$---a contradiction.
\qed
\end{proof}
\fi

\ifshort
\begin{lemma} \label{lem:23}
For $i \in [3k]$, $R_i$ is not a \pc{} in $C_P$ and $T_i$ does not contain a lowest \pc{} of size less than $s/4$.
\end{lemma}
\fi

\ifshort
\begin{lemma}
\label{lem:4}
For $i \in [3k]$, $L_i$ is the only \pc{} contained in $T_i$, and the subtree of $T$ induced by $(V(T_i)-V(L_i)) \cup \{p_i\}$ is a \pc{} of size $s$.
\end{lemma}

\begin{proof}
Since $|T_i| > s$, any $T_i$ must contain at least one \pc{}. Since $C_P$ has $6k$ \pcs{}, at least one of the $3k$ $T_i$'s contains at most one \pc{}, because otherwise the \pcs{} containing vertices in $H$ are not accounted for. Therefore, at least one $T_i$ contains exactly one \pc{} $C$, which must be a lowest \pc{} in $T_i$. By Lemma~\ref{lem:23}, $C \neq R_i$ and $|C| \geq s/4$, and hence $C$ cannot be any proper subtree of $L_i$, $R^l_i$, or $R^l_i$. This leaves $L_i$, $R^l_i$, and $R^l_i$ as the only possible choices for $C$.

Suppose that $C = R^l_i$. After removing $C$, the partial-$T_i$, denoted $T^-_i$, has size $s-1+a_i-(s/2-1)=s/2 + a_i$, and contains no \pcs{}. Let $D$ be the set of vertices that are not in $T^-_i$, and are in the same group as $T^-_i$. Observe that for any $j\neq i, j \in [3k]$, if a  vertex in $L_j$ is in $D$ then all vertices in $L_j$ are in $D$. This is true because, by Lemma~\ref{lem:23}, {\em all} vertices in $L_j$ belong to the same \pc{}; otherwise $L_j$ would have a lowest \pc{} of size less than $s/4$. This means that the \pc{} containing $T^-_i$ has size $|T^-_i|+|D| \geq s/2 + a_i + a_j > s$. Therefore, $D$ does not include any vertex in $L_j$. Similarly, $D$ does not include any vertex in $R^l_j$ or $R^r_j$. It follows that $D$ consists only of vertices in $H$, the roots of the $T_i$'s, and the roots of the $R_i$'s, $i \in [3k]$. However, there are only $9k$ such vertices, which means $|D| \leq 9k$ and hence the \pc{} containing $T^-_i$ has size $|T^-_i|+|D| \leq s/2 + a_i + 9k < s$. The last inequality is true because $a_i \leq s/2-18k$, for $i \in [k]$.

Therefore, $C \neq R^l_i$. By a similar argument, $C \neq R^r_i$. It follows that $C = L_i$. After $L_i$ is removed, the resulting partial-$T_i$ along with $p_i$ in $H$ induces a subtree $C_i$ of size exactly $s$, and hence must be a \pc{} by itself.

After both $L_i$ and $C_i$ are removed, there are $(3k-1)$-many $T_i$'s and $6k-2$ \pcs{} remaining in $T$. Thus, there is at least one $T_j$ containing exactly one \pc{}. By the same argument above, the only \pc{} contained in $T_j$ is $L_j$. Repeating this argument $3k$ times in total proves the lemma.
\qed
\end{proof}
\fi

\iflong\begin{lemma}
\label{lem:4}
For $i \in [3k]$, $L_i$ is the only \pc{} contained in $T_i$, and the subtree of $T$ induced by $(V(T_i)-V(L_i)) \cup \{p_i\}$ is a \pc{} of size $s$.
\end{lemma}

\begin{proof}
Since $|T_i| > s$, any $T_i$ must contain at least one \pc{}. Since $C_P$ has $6k$ \pcs{}, at least one of the $3k$ $T_i$'s contains at most one \pc{}, because otherwise the \pcs{} containing vertices in $H$ are not accounted for. Therefore, at least one $T_i$ contains exactly one \pc{} $C$, which must be a lowest \pc{} in $T_i$. By Lemma~\ref{lem:2}, $C \neq R_i$. By Lemma~\ref{lem:3}, $|C| \geq s/4$, and hence $C$ cannot be any proper subtree of $L_i$, $R^l_i$, or $R^l_i$. This leaves $L_i$, $R^l_i$, and $R^l_i$ as the only possible choices for $C$.

Suppose that $C = R^l_i$. After removing $C$, the partial-$T_i$, denoted $T^-_i$, has size $s-1+a_i-(s/2-1)=s/2 + a_i$, and contains no \pcs{}. Let $D$ be the set of vertices that are not in $T^-_i$, and are grouped together with $T^-_i$. Observe that for any $j\neq i, j \in [3k]$, if a  vertex in $L_j$ is in $D$ then all vertices in $L_j$ are in $D$. This is true because, by Lemma~\ref{lem:3}, {\em all} vertices in $L_j$ belong to the same \pc{}; otherwise $L_j$ would have a lowest \pc{} of size less than $s/4$. This means that the \pc{} containing $T^-_i$ has size $|T^-_i|+|D| \geq s/2 + a_i + a_j > s$. Therefore, $D$ does not include any vertex in $L_j$. similarly, $D$ does not include any vertex in $R^l_j$ or $R^r_j$. It follows that $D$ consists only of vertices in $H$, the roots of the $T_i$'s, and the roots of the $R_i$'s, $i \in [3k]$. However, there are only $9k$ such vertices, which means that the \pc{} containing $T^-_i$ has size $|T^-_i|+|D| \leq s/2 + a_i + 9k < s$. The last inequality is true because, by Assumption~\ref{assumption:ass}, $a_i \leq s/2-18k$ (assuming, without loss of generality, that $k >0$), for $i \in [k]$.

Therefore, $C \neq R^l_i$. By a similar argument, $C \neq R^r_i$. It follows that $C = L_i$. After $L_i$ is removed, the resulting partial-$T_i$ along with $p_i$ in $H$ induces a subtree $C_i$ of size exactly $s$, and hence must be a \pc{} by itself.

After both $L_i$ and $C_i$ are removed, there are $(3k-1)$-many $T_i$'s and $6k-2$ \pcs{} remaining in $T$. Thus, there is at least one $T_j$ containing exactly one \pc{}. By the same argument above, the only \pc{} contained in $T_j$ is $L_j$. Repeating this argument $3k$ times in total proves the lemma.
\qed
\end{proof}
\fi

\ifshort
Lemma~\ref{lem:4} shows that, in a solution $P$ of $(T, 6k-1, 4k)$, each of $L_1, \ldots, L_{3k}$ is a \pc{}, and the remaining part of each of the $3k$ $T_i$'s is a \pc{} whose size is $s$.
Based on this, the theorem below easily follows:
\fi
\begin{theorem} \label{thm:nphardness} \threebtp{} is~\NP-complete.
\end{theorem}
\iflong
\begin{proof}
It is easy to see that \threebtp{} is in \NP.

To show \NP-hardness, we reduce from {\sc 3-Partition}. It is well known that {\sc 3-Partition} is \NP-hard in the strong sense~\cite{gj}. So we can restrict our attention to the instances of {\sc 3-Partition} in which the numbers in the instance are all bounded by a polynomial in the instance length.
Consider the reduction from {\sc 3-Partition} to \threebtp{} that maps each instance $(S=\langle a_1, \ldots, a_{3k} \rangle, s)$ of {\sc 3-Partition} to the instance $(T, 6k-1, b=4k)$. This reduction is computable in polynomial time.

Clearly, if  $(S=\langle a_1, \ldots, a_{3k} \rangle, s)$ is a yes-instance of {\sc 3-Partition}  then we can construct a solution $P$ of $(T, 6k-1, 4k)$ by cutting each $L_i$, $i \in [3k]$, plus the $3k-1$ edges of the path $H$, for a total of $6k-1$ edges. A grouping of the resulting \pcs{} follows trivially from the 3-partitioning of the $a_i$'s. To prove the converse, let $P$ be a solution to the instance $(T, 6k-1, 4k)$. By Lemma~\ref{lem:4}, $C_P$ consists of $\{L_1, \ldots, L_{3k}\}$ and $3k$ additional \pcs{}, each of size $s$. Since each group has size $s$, each of the $3k$ \pcs{} must occupy a group by itself, and the remaining $L_1, \ldots, L_{3k}$ components of sizes $a_1, \ldots, a_{3k}$, respectively, must be assigned to the remaining $k$ groups. Since each number in
$a_1, \ldots, a_{3k}$ satisfies $s/4 < a_i < s/2$, each group is assigned exactly three numbers, and
$\{a_1, \ldots, a_{3k}\}$ is a yes-instance of {\sc 3-Partition}. This completes the proof.
\qed
\end{proof}
\fi
\iflong
\section{\gtp~and \btp~are $W[1]$-complete}
\label{sec:whardness}
\fi
\ifshort
\section{$W[1]$-completeness}
\label{sec:whardness}
\fi
\iflong
In this section, we show that \gtp~and \btp~are $W[1]$-complete. To show their membership in $W[1]$, we use a characterization of the class $W[1]$ given in~\cite{yijia}. We start by showing that \gtp~is $W[1]$-hard. We then show that \gtp~and \btp~are equivalent modulo \fpt-reducibility, which implies the $W[1]$-hardness of \btp.   \\

To show that \gtp~is $W[1]$-hard, we provide an \fpt-reduction from the $W[1]$-complete  \textsc{$k$-Multi-Colored Clique ($k$-MCC)} problem (\cite{fptbook,viallete}) defined as follows: Given a graph $M = (V(M),E(M))$ and a proper $k$-coloring of the
vertices $f : V(M) \longrightarrow C$, where $C = \{1,2,..., k\}$ and each color class has the same cardinality, decide whether there exists a clique $Q \subseteq V(M)$ of size $k$ such that, $\forall u, v \in Q$, $f(u) \ne f(v)$. For $i \in [k]$, we define $C_i =\{v \in M \mid f(v)=i\}$ to be the color class  consisting of all vertices whose color is $i$.
\fi
\ifshort
To show that \gtp~is $W[1]$-hard (membership in $W[1]$ is shown using a characterization of the class $W[1]$ given in~\cite{yijia}), we provide a fixed-parameter tractable reduction (\fpt-reduction) from the $W[1]$-complete  \textsc{$k$-Multi-Colored Clique ($k$-MCC)} problem (\cite{fptbook,viallete}) defined as follows: Given a graph $M = (V(M),E(M))$ and a proper $k$-coloring of the
vertices $f : V(M) \longrightarrow C$, where $C = \{1,2,..., k\}$ and each color class has the same cardinality, decide whether there exists a clique $Q \subseteq V(M)$ of size $k$ such that, $\forall u, v \in Q$, $f(u) \ne f(v)$. For $i \in [k]$, we define $C_i =\{v \in M \mid f(v)=i\}$ to be the color class  consisting of all vertices whose color is $i$.
\fi
Let $n = |C_i|$, $i \in [k]$, and let $N=k\cdot n$. We label the vertices in $C_i$ arbitrarily as $v_{1}^{i}, \ldots, v_{n}^{i}$. We first introduce some terminologies.

For a finite set $X \subseteq \nat$ and $\ell \in \integ$, we say that $X$ is \emph{$\ell$-non-averaging} if for any $\ell$ numbers $x_1, \ldots, x_{\ell} \in X$, and for any number $x \in X$, the following holds: if $x_1 + \cdots +x_{\ell} = \ell \cdot x$ then $x_1 = \cdots = x_{\ell} =x$.

Let $X=\{x_1, \ldots, x_n\}$ be a $(k-1)$-non-averaging set. It is known that we can construct such a set $X$ such that each element $x_i \in X$, $i \in [n]$, is polynomial in $n$ (for instance, see~\cite{nonaveraging}). Jensen \etal~\cite{binpacking} showed that a $(k-1)$-non-averaging set of cardinality $n$, in which each number is at most $k^2n^2 \leq n^4$, can be constructed in polynomial time in $n$; we will assume that $X$ is such a set. Let $k'= k + {k \choose 2}$, and let $z=k'^2n^5$. Choose $2k$ numbers $b_1, \ldots, b_k, c_1, \ldots, c_k \in \nat$ such that $b_j=k'^{2j} \cdot z$ for $j \in [k]$, and $c_j =k'^{2(k+j)}\cdot z$ for $j \in [k]$.  Observe that each number in the sequence $b_1, \ldots, b_k, c_1, \ldots, c_k$ is equal to the preceding number multiplied by $k'^2$, and that the smallest number $b_1$ in this sequence is $k'^2 \cdot z \geq k'^4 n^5$. For each $j, j' \in [k], j < j'$, we choose a number $c_{j}^{j'}=c_k\cdot k'^{2((j-1)k -j(j-1)/2 +j'-j)}$. That is, each number in the sequence $c_{1}^{2}, \ldots, c_{1}^{k}, c_{2}^{3}, \ldots, c_{2}^{k}, \ldots, c_{k-1}^{k}$ is equal to the preceding one multiplied by $k'^2$, and the smallest number $c_{1}^{2}$ in this sequence is equal to $k'^2 \cdot c_k$.

We construct a tree $T$ rooted at a vertex $r$ as follows. For a vertex $v_{i}^{j}$, $i \in [n], j \in [k]$, we correspond a {\em vertex-gadget} (for vertex  $v_{i}^{j}$) that is a star $S_{v_{i}^{j}}$ with $c_j - (k-1)b_j - (k-1)x_i-1$ leaves, and hence with $c_j - (k-1)b_j - (k-1)x_i$ vertices; we label the root of the star $r_{v_{i}^{j}}$, and add the edge $rr_{v_{i}^{j}}$ to $T$. See Figure~\ref{fig:vertexgadget} for illustration. For each edge $e$ in $M$ between two vertices $v_{i}^{j}$ and $v_{p}^{q}$, $i,p \in [n], j, q \in [k], j < q$, we create two stars $S'_{v_{i}^{j}}$ and $S'_{v_{p}^{q}}$, with $b_j + x_i-1$ and $b_q+x_p-1$ leaves, respectively, and of roots $r'_{v_{i}^{j}}$ and $r'_{v_{p}^{q}}$, respectively. We introduce a star $S_e$ with root $r_e$ and $c_{j}^{q}-1$ leaves, and connect $r_e$ to $r'_{v_{i}^{j}}$ and $r'_{v_{p}^{q}}$ to form a tree $T_e$ with root $r_e$ that we call an {\em edge-gadget} (for edge $e$). We connect $r_e$ to $r$. See Figure~\ref{fig:edgegadget} for illustration. Note that the number of vertices in $T_e$ that are not in $S'_{v_{i}^{j}} \cup S'_{v_{p}^{q}}$ is exactly $c_{j}^{q}$. Finally, we create $k'+1$ copies of a star $S_{fix}$ consisting of $c_{k-1}^{k} +k'+1$ many vertices, and connect the root $r$ of $T$ to the root of each of these copies. This completes the construction of $T$. Let $t=|T|$. We define the reduction from \mcc~to \gtp~to be the map that takes an instance $I=(M, f)$ of \mcc~and produces the instance $I'=(T, k', b=k+{k \choose 2}, c_1, \ldots, c_k, c_{1}^{2}, \ldots, c_{1}^{k}, c_{2}^{3}, \ldots, c_{2}^{k} \ldots, c_{k-1}^{k}, t')$, where $k'=k+3{k \choose 2}$ and $t'=t - \sum_{j=1}^{k} c_j - \sum_{j, q \in [k], j < q}c_{j}^{q}$. Clearly, this reduction is an \fpt-reduction. Next, we describe the intuition behind this reduction.

\begin{figure}[htbp]
\centering
\begin{minipage}{.4\textwidth}
\begin{overpic}[clip, trim=7cm 8cm 5cm -1.6cm, width=1.00\textwidth, tics=10]{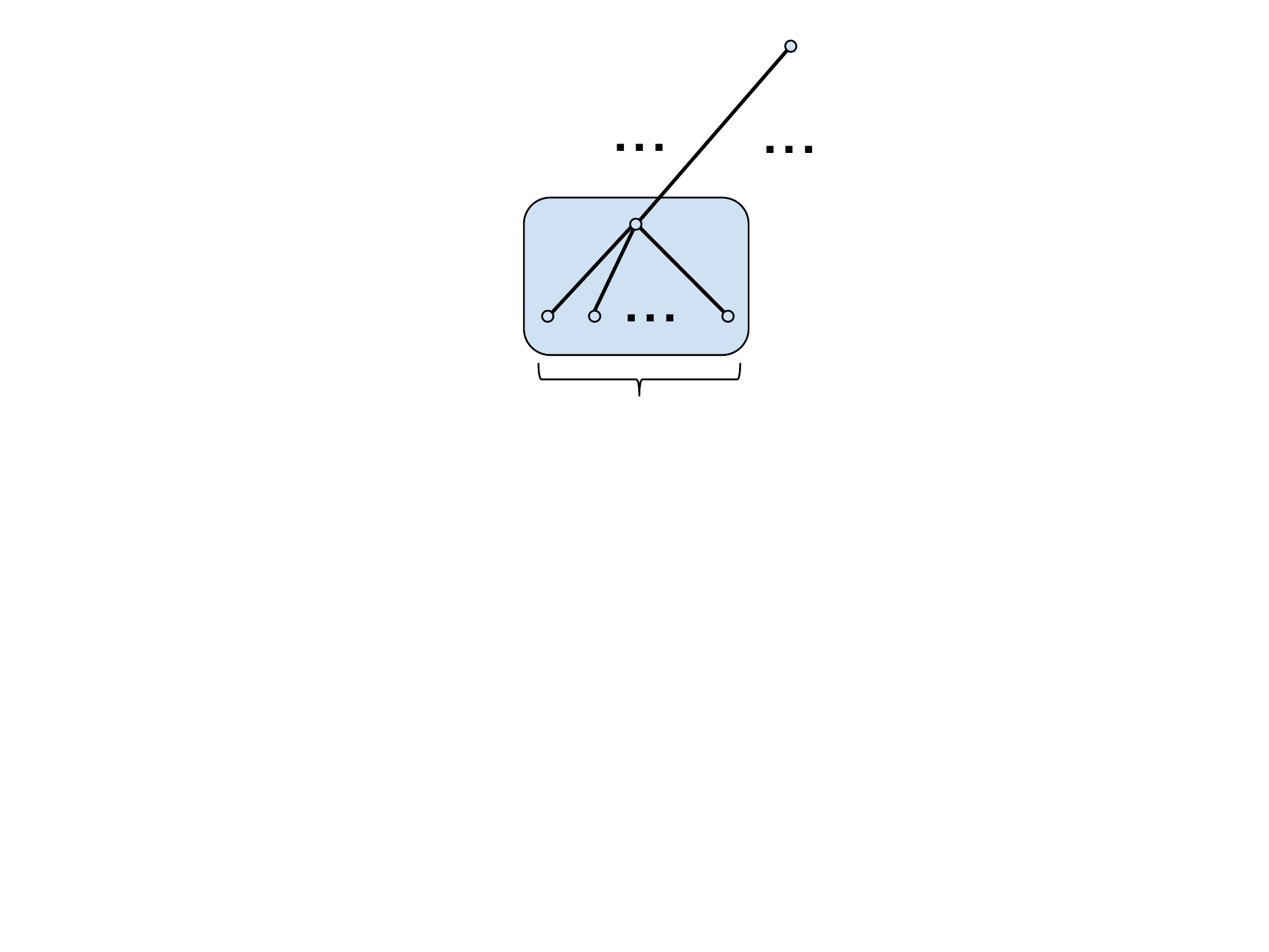}

 \put (61,75) {$r$}
 \put (15,48) {$S_{v_{i}^{j}}$}
 \put (5,-1) {$c_j-(k-1)b_j-(k-1)x_i-1$}
\end{overpic}
\caption{Illustration of the vertex-gadget for $v_{i}^{j}$.\label{fig:vertexgadget}}
\end{minipage}%
\hspace{0.2cm}
\begin{minipage}{.53\textwidth}
\begin{overpic}[clip, trim=5.5cm 5.5cm 4cm 0.6cm, width=1.00\textwidth,tics=10]{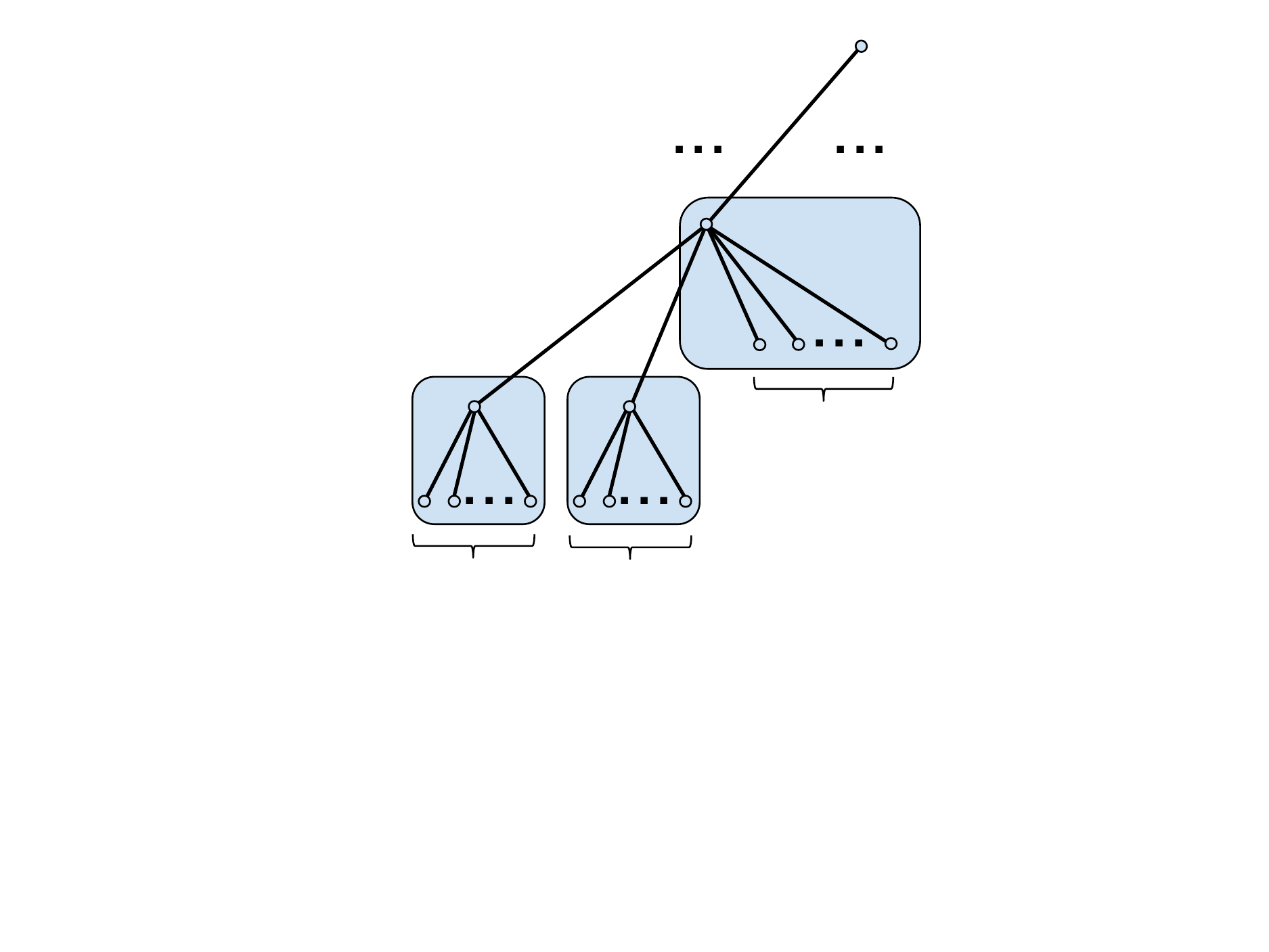}

 \put (78,88) {$r$}
  \put (67,64) {$S_e$}
   \put (59,55) {$r_e$}
  \put (62,25) {$c_{j}^{q}-1$}

  \put (13,39) {$S'_{v_{i}^{j}}$}
   \put (8,28) {$r'_{v_{i}^{j}}$}
    \put (32,28) {$r'_{v_{p}^{q}}$}
  \put (36,39) {$S'_{v_{p}^{q}}$}
  \put (4,0) {$b_j+x_i-1$}
  \put (33,0) {$b_q+x_p-1$}
\end{overpic}
\caption{Illustration of the edge-gadget for $e=v_{i}^{j}v_{p}^{q}$.\label{fig:edgegadget}}
\end{minipage}
\end{figure}

Each number $c_j$, $j \in [k]$, chosen above, will serve as a ``signature'' for class $C_j$, in the sense that it will ensure that in any solution to the instance, a vertex-gadget corresponding to a vertex in class $C_j$ is ``cut'' and placed in the group of size $c_j$. Each number $c_{j}^{j'}$,  $j, j' \in [k], j < j'$, will serve as a ``signature'' for the class-pair $(C_j, C_{j'})$, in the sense it will ensure that in a solution exactly one edge-gadget corresponding to an edge $e$ between classes $C_j$ and $C_j'$ is cut and the star $S_e$ is placed in the group whose size is $c_{j}^{j'}$. Each number $b_j$, $j \in [k]$, will serve as a ``signature'' for any edge such that one of its endpoints is in $C_j$ (\ie, a signature for an arbitrary vertex in $C_j$), ensuring that in a solution, $k-1$ of these edges are cut. Finally, the choice of the $x_i$'s, for $i \in [n]$, to be elements of a $(k-1)$-non-averaging set, will ensure that all the edges cut that are incident to vertices in the same class $C_j$, $j \in [k]$, are incident to the same vertex in $C_j$.

Next, we prove the correctness of the reduction. One direction is easy:

\ifshort
\begin{lemma}
\label{lem:togtp}
If $(M, f)$ is a yes-instance of \mcc~then $I'$ is a yes-instance of \gtp.
\end{lemma}
\fi

\iflong
\begin{lemma}
\label{lem:togtp}
If $(M, f)$ is a yes-instance of \mcc~then $I'$ is a yes-instance of \gtp.
\end{lemma}

\begin{proof}
If $(M, f)$ has a clique $Q$ of size $k$, then we can form a solution $P=(E_P, \f)$ for \gtp~as follows. For every vertex $v_{i}^{j} \in Q$, $E_P$ contains the edge between $r$ and the root $r_{v_{i}^{j}}$ of the vertex-gadget (star) $S_{v_{i}^{j}}$; $\f$ assigns the vertex-gadget $S_{v_{i}^{j}}$ to group $G_j$. For each edge $e=v_{i}^{j}v_{p}^{q}$, $j,q \in [k], j < q$, between two vertices in $Q$, $P$ cuts the edge between the root $r_e$ of $T_e$ and $r$, and cuts the two edges between $r_e$ and its two children $r'_{v_{i}^{j}}$ and $r'_{v_{p}^{q}}$; $\f$ assigns the star $S'_{v_{i}^{j}}$ to group $G_j$, the star $S'_{v_{p}^{q}}$ to group $G_q$, and the rest of $T_e$ (\ie, $S_e$) to group $G_{j}^{q}$. The remainder of $T$ is assigned to the last group $G_{rest}$. It is easy to verify that $P$ cuts exactly $k'$ edges in $T$. Moreover, the number of vertices assigned by $\f$ to each group $G_j$, $j \in [k]$, is exactly $c_j$; the number of vertices assigned to each group $G_{j}^{q}$, $j,q \in [k], j < q$, is exactly $c_{j}^{q}$; and the number of vertices assigned to $G_{rest}$ is exactly $t'=t - \sum_{j=1}^{k} c_j - \sum_{j, q \in [k], j < q}c_{j}^{q}$.  It follows that $P$ is a solution to $I'$.
\qed
\end{proof}

Next, we prove the converse. Let $P=(E_P, \f)$ be a solution to the instance $I'=(T, k', b=k+{k \choose 2}, c_1, \ldots, c_k, c_{1}^{2}, \ldots, c_{1}^{k}, c_{2}^{3}, \ldots, c_{2}^{k} \ldots, c_{k-1}^{k}, t')$ of \gtp. Let $G_j$, $j \in [k]$, denote the group of size $c_j$, $G_{j}^{q}$, $j, q \in [k], j < q$, denote the group of size $c_{j}^{q}$, and $G_{rest}$ denote the group of size $t'$. First, observe the following:

\begin{observation} \label{obs:root}
 If $P=(E_P, \f)$ is a solution for $I'$, then $\f$~assigns the root $r$ of $T$ to group $G_{rest}$.
\end{observation}

\begin{proof}
Since $P$ cuts at most $k'$ edges, at least one substar of a copy of the star $S_{fix}$ must be connected to $r$ in $T - E_P$. Since the number of vertices in $S_{fix}$ is $c_{k-1}^{k} +k'+1$, the substar of $S_{fix}$ that is connected to $r$ must have size greater that $c_{k-1}^{k}$, where $c_{k-1}^{k}$ is the size of the largest group other than $G_{rest}$. Therefore, this substar, and hence $r$, must be assigned to group $G_{rest}$ by $\f$. \qed
\end{proof}

\begin{lemma}
\label{lem:vertexgadget}
For each $j \in [k]$, $P$ cuts exactly one edge between the root $r$ of $T$ and the root of a vertex-gadget corresponding to a vertex from class $C_j$ in $M$. Moreover, $\f$ assigns this substar to group $G_j$.
\end{lemma}

\begin{proof}
Let $j \in [k]$. We first claim that one of the components in $T-E_P$ assigned to group $G_j$ by $\f$ must be a substar of a vertex-gadget $S_{v_{i}^{j}}$, for some $i \in [n]$. To prove the statement of the lemma, we show that the only way that the sizes of the components placed in $G_j$ could add up to $c_j$ is that if a substar of $S_{v_{i}^{j}}$ is assigned to $G_j$ by $\f$.

First, observe that no substar of $S_{v_{i'}^{q}}$, $i' \in [n]$, $q > j$, could be placed in $G_j$ for the following reason. The star $S_{v_{i'}^{q}}$ has size $c_q-(k-1)b_q-(k-1)x_{i'}$.  By cutting at most $k'$ star-edges of $S_{v_{i'}^{q}}$, the size of $S_{v_{i'}^{q}}$ can be reduced by at most $k'$, and hence, the size of any substar of $S_{v_{i'}^{q}}$ assigned to group $G_j$ is at least: $c_q-(k-1)b_q-(k-1)x_{i'} - k' > c_q -kb_q -kx_{i'} -k'$. Since  $x_{i'} \leq n^4 < b_1 \leq b_q$ and since $k' < b_1 \leq b_q$, the size of any substar of $S_{v_{i'}^{q}}$ assigned to group $G_j$ is at least $c_q -3kb_q > c_q -k'b_q > c_j$ because $j < q$ and $c_q -c_j \geq k'b_q$ by the choice of the $c_1, \ldots, c_k$.

Second, observe that, by the same argument as above, no subtree of an edge-gadget $T_e$ that contains the root $r_e$ of $T_e$ can belong to $G_j$ simply because such a tree contains at least $c_{p}^{q}-k'$ many vertices, and this number exceeds $c_j$, for any $j, p, q \in [k], p < q$ by the choice of the numbers $c_j$ and $c_{p}^{q}$. Similarly, no substar of a copy of the star $S_{fix}$ can be in $G_j$.

It follows from above that, if no substar of a vertex-gadget $S_{v_{i}^{j}}$ is in $G_j$, then the largest number of vertices that can be placed in $G_j$ is obtained by placing $k'$ vertex-gadgets corresponding to vertices in group $j-1$, which would result in a number of vertices in $G_j$ that is at most $k'c_{j-1} < c_j$, by the choice of the numbers $c_1, \ldots, c_k$.

Therefore, a substar of $S_{v_{i}^{j}}$ must be assigned by $\f$~to group $G_j$. Since by Observation~\ref{obs:root} the root $r$ of $T$ is assigned to $G_{rest}$ by $\f$, the edge $rr_{v_{i}^{j}}$ must be cut by $P$. This completes the proof.
\qed
\end{proof}

\begin{lemma}
\label{lem:edgegadget}
For each $j, q \in [k]$, $j < q$, $P$ cuts exactly one edge $rr_e$, between the root $r$ of $T$ and the root $r_e$ of an edge-gadget $T_e$ corresponding to an edge between a vertex in color class $C_j$ and a vertex in color class $C_q$. Moreover, $\f$ assigns a substar of $S_e$ (in $T_e$) to group $G_{j}^{q}$.
\end{lemma}

\begin{proof}
The proof follows similar arguments to that of Lemma~\ref{lem:vertexgadget}, by analyzing the components in $T-E_P$ assigned to group $G_{j}^{q}$.

First, observe that $G_{j}^{q}$ must contain a subtree from some edge-gadget. This is because the largest vertex-gadget has size smaller than $c_k$, and $c_k < k'^2 c_{1}^{2}$, where $c_{1}^{2}$ is the smallest size of any group $G_{j}^{q}$, $j, q \in [k]$, $j < q$. Therefore, cutting $k'$ vertex-gadgets and placing them all in a group $G_{j}^{q}$ cannot result in $G_{j}^{q}$ having size $c_{j}^{q}$. Moreover, no substar of a copy of $S_{fix}$ can be placed in $G_{j}^{q}$ because the size of $S_{fix}$ is $c_{k-1}^{k}+k'+1$. Therefore, $G_{j}^{q}$ must contain a subtree of an edge-gadget.

Second, observe that if no substar of a star $S_{e'}$, contained in an edge gadget $T_{e'}$, is placed in $G_{j}^{q}$, then the components placed in $G_{j}^{q}$ consist only of substars of vertex-gadgets and of substars of stars $S'_{v_{i'}^{j'}}$ that are contained in edge-gadgets (plus isolated vertices). Each substar of a vertex gadget has size at most $c_k$ and each substar of a star $S'_{v_{i'}^{j'}}$ has size at most $b_k+n^4 < c_k$. Therefore, $k'$ such substars have size at most $k'\cdot c_k < c_{j}^{q}$.

Third, $G_{j}^{q}$ cannot contain a substar of a star $S_{e'}$, contained in an edge-gadget $T_{e'}$, such that $e'$ has an endpoint $v_{j'}^{q'}$ in a class $C_{q'}$, where either $q' > q$ or $q'=q$ and $j' > j$. This is because such a substar, even after cutting $k'$ of its edges, will have size at least $c_{j'}^{q'} -k' > k'^2 \cdot c_{j}^{q} -k' > c_{j}^{q}$, by the choice of the numbers $c_{1}^{2}, \ldots, c_{k-1}^{k}$.

Finally, the substars of stars $S_{e'}$ of edge-gadgets $T_{e'}$ contained in $G_{j}^{q}$ cannot all correspond to edges $e'$ both of whose endpoints are in classes $C_{j'}$, $C_{q'}$ where $j' < q'$ satisfy $q' < q$ or $q'=q$ and $j' < j$. This is because $k'$ such substars will have size at most $k' \cdot c_{j'}^{q'} < c_{j}^{q}$, by the choice of $c_{1}^{2}, \ldots, c_{k-1}^{k}$.

It follows from above that $G_{j}^{q}$ must contain a substar of a star $S_e$ such that $e=v_{i}^{j}v_{p}^{q}$, where vertex $v_{i}^{j}$ corresponds to a vertex from class $C_j$, and vertex $v_{p}^{q}$ corresponds to a vertex from class $C_q$. Since by Observation~\ref{obs:root} the root $r$ of $T$ is assigned to $G_{rest}$ by $\f$, the edge $rr_e$ must be cut by $P$. \qed
\end{proof}

\begin{lemma}
\label{lem:edgegadgetexact}
For each $j, q \in [k]$, $j < q$, there is a solution $P$ that such that $\f$ assigns group $G_{j}^{q}$ a single component in $T-E_P$, consisting of a (whole) star $S_e$ in edge-gadget $T_e$ that corresponds to an edge $e$ between a vertex in color class $C_j$ and a vertex in color class $C_q$.
\end{lemma}

\begin{proof}
By Lemma~\ref{lem:edgegadget}, we know that there is a solution $P$ such that $\f$ assigns to each group $G_{j}^{q}$, $j, q \in [k], j < q$, a substar of a star $S_e$ contained in an edge-gadget $T_e$ corresponding to an edge $e=v_{i}^{j}v_{p}^{q}$, between a vertex $v_{i}^{j} \in C_j$ and a vertex $v_{p}^{q} \in C_q$. Choose a solution $P$ that minimizes the number of edges it cuts from the stars $S_{e'}$, contained in edge-gadgets $T_{e'}$, and placed in the groups $G_{j}^{q}$, $j, q \in [k]$, $j < q$. We claim that $P$ satisfies the statement of the lemma.

Suppose not, and let $G_{j}^{q}$, for some $j, q \in [k]$, $j < q$, be a group that is assigned a proper substar of $S_e$, where $e=v_{i}^{j}v_{p}^{q}$, $i, p \in [n]$. Then $P$ must cut a star-edge $r_ew$ of $S_e$. First, observe that $G_{j}^{q}$ cannot contain a component $K$ that is not contained in some substar of a star $S_{e'}$ (of an edge-gadget $T_{e'}$) that is assigned by $\f$ to some group $G_{j'}^{q'}$, for some $j', q' \in [k], \{j', q'\} \neq \{j, q\}$. Otherwise, we can modify $P$ so that $P$ does not cut the edge $r_ew$, and cuts instead an edge incident to a leaf $x$ in $K$, if $K$ itself is not a leaf (otherwise, we do not cut anything), and modify $\f$ to swap $w$ with $u$. By modifying $P$ as such, we would obtain another solution that cuts fewer edges from the stars $S_e$'s, whose substars are assigned to the groups, contradicting the choice of $P$. It follows from above that each group $G_{j}^{q}$ consists of a substar of $S_e$, for some edge $e=v_{i}^{j}v_{p}^{q}$, plus components that are contained in stars $S_{e'}$ of edge-gadgets $S_{e'}$, such that a substar of $S_{e'}$ is assigned to some group $G_{j'}^{q'}$, for some $j', q' \in [k], \{j', q'\} \neq \{j, q\}$. Moreover, no such component that is contained in a star $S_{e'}$ can contain the root $r_{e'}$ of $S_{e'}$. Otherwise, the size of that component would exceed $k'$, and hence, this component, together with the substar of $S_e$ (whose size is at least $c_{j}{q} - k'$ would result in a total size that exceeds the size $c_{j}^{q}$ of group $G_{j}^{q}$. It follows that group $G_{j}^{q}$ consists of a proper substar of $S_e$, plus isolated vertices, where each is a leaf in some $S_{e'}$ (of an edge-gadget $T_{e'}$), such that a proper star of $S_{e'}$ is assigned by $\f$ to some other group $G_{j'}^{q'}$.

Now construct the following auxiliary graph, each of whose vertices corresponds to a group $G_{j}^{q}$ that contains a proper substar of some $S_e$ in an edge-gadget $T_e$, and in which there is a directed edge from $G_{j}^{q}$ to $G_{j'}^{q'}$ if the proper substar of $S_e$ assigned to $G_{j}^{q}$ contains a leaf of a star $S_{e'}$ such that a proper substar of $S_{e'}$ is in $G_{j'}^{q'}$. By the definition of the vertex-set of this auxiliary graph, and from the above discussion, each vertex in this auxiliary graph must have out-degree at least 1. Therefore, there must exist a cycle in this auxiliary graph. Such a cycle, however, would clearly contradict the choice of $P$, as we can define another solution that restores an edge from each of the proper substars contained in the groups of this cycle, without affecting the size of each group in this cycle.
\qed
\end{proof}

\begin{corollary}\label{cor:exactassignment}
There is a solution $P$ that cuts exactly $k'=k + 3{k \choose 2}$ edges from $T$ as follows. For each $j \in [k]$, $P$ cuts exactly one edge between the root $r$ of $T$ and the root of a vertex-gadget corresponding to a vertex in color class $C_j$; moreover, $\f$ assigns the resulting vertex-gadget to group $G_j$. For each $j, q \in [k]$, $j < q$, $P$ cuts exactly 3 edges from one edge-gadget $T_e$, corresponding to an edge $e$ between a vertex $v_{i}^{j}, i \in [n]$, in color classes $C_j$, and a vertex $v_{p}^{q}, p \in [n]$, in color class $C_q$; those 3 edges are the edges $rr_e$, $r_er'_{v_{i}^{j}}$, and $r_er'_{v_{p}^{q}}$, where $r_e$ is the root of star $S_e$ in $T_e$, and $r'_{v_{i}^{j}}, r'_{v_{p}^{q}}$ are the roots of stars $S'_{v_{i}^{j}}, S'_{v_{p}^{q}}$ in $T_e$, respectively; moreover, $\f$ assigns $S_e$ to group $G_{j}^{q}$.
\end{corollary}
\fi

\ifshort
To prove the converse, let $P=(E_P, \f)$ be a solution to the instance $I'=(T, k', b=k+{k \choose 2}, c_1, \ldots, c_k, c_{1}^{2}, \ldots, c_{1}^{k}, c_{2}^{3}, \ldots, c_{2}^{k} \ldots, c_{k-1}^{k}, t')$ of \gtp. Let $G_j$, $j \in [k]$, denote the group of size $c_j$, $G_{j}^{q}$, $j, q \in [k], j < q$, denote the group of size $c_{j}^{q}$, and $G_{rest}$ denote the group of size $t'$. We have:

\begin{lemma}\label{cor:exactassignment}
There is a solution $P$ that cuts exactly $k'=k + 3{k \choose 2}$ edges from $T$ as follows. For each $j \in [k]$, $P$ cuts exactly one edge between the root $r$ of $T$ and the root of a vertex-gadget corresponding to a vertex in color class $C_j$; moreover, $\f$ assigns the resulting vertex-gadget to group $G_j$. For each $j, q \in [k]$, $j < q$, $P$ cuts exactly 3 edges from one edge-gadget $T_e$, corresponding to an edge $e$ between a vertex $v_{i}^{j}, i \in [n]$, in color classes $C_j$, and a vertex $v_{p}^{q}, p \in [n]$, in color class $C_q$; those 3 edges are the edges $rr_e$, $r_er'_{v_{i}^{j}}$, and $r_er'_{v_{p}^{q}}$, where $r_e$ is the root of star $S_e$ in $T_e$, and $r'_{v_{i}^{j}}, r'_{v_{p}^{q}}$ are the roots of stars $S'_{v_{i}^{j}}, S'_{v_{p}^{q}}$ in $T_e$, respectively; moreover, $\f$ assigns $S_e$ to group $G_{j}^{q}$.
\end{lemma}
\fi

\iflong
\begin{proof}
By Lemma~\ref{lem:vertexgadget}, any solution $P$
satisfies that, for each $j \in [k]$, $P$ cuts exactly one edge $rr_{v_{i}^{j}}$ between the root $r$ of $T$ and the root $r_{v_{i}^{j}}$ of a vertex-gadget $S_{v_{i}^{j}}$ corresponding to a vertex from class $C_j$ in $M$, and $\f$ assigns a substar of  $S_{v_{i}^{j}}$ to group $G_j$. By Lemma~\ref{lem:edgegadgetexact}, we can assume that, for each $j, q \in [k]$, $j < q$, $\f$ assigns group $G_{j}^{q}$ a single component consisting of a star $S_e$ of an edge-gadget $T_e$, corresponding to an edge $e=v_{i}^{j}v_{p}^{q}$. Because the size of $G_{j}^{q}$ is exactly the size of $S_e$,  $P$ must cut the 3 edges $rr_e$, $r_er'_{v_{i}^{j}}$, and $r_er'_{v_{p}^{q}}$ to separate $S_e$ from the rest of $T_e$. Since $|E_P|$ is at most $k'=k + 3{k \choose 2}$, it follows that the above edges are all the edges of $T$ that are cut by $P$, and hence, the substar of $S_{v_{i}^{j}}$ assigned by $\f$ to group $G_j$, is indeed the whole star $S_{v_{i}^{j}}$. \qed
\end{proof}

We are now ready to prove the converse of Lemma~\ref{lem:togtp}:
\fi
\begin{lemma}\label{lem:converse}
If $I'$ is a yes-instance of \gtp~then $(M, f)$ is a yes-instance of $k$-MCC.
\end{lemma}

\begin{proof}
 By \ifshort Lemma~\ref{cor:exactassignment}\fi \iflong Corollary~\ref{cor:exactassignment}\fi, we can assume that $I'$ has a solution $P=(E_P, \f)$ that cuts $k+3{k \choose 2}$ edges, and that satisfies the properties in the \ifshort lemma\fi \iflong corollary\fi. Let $rr_{v_{i_1}^{1}}, \ldots, rr_{v_{i_k}^{k}}$, $i_1, \ldots, i_k \in [n]$, be the edges between the root $r$ of $T$ and the roots of the vertex-gadgets $S_{v_{i_1}^{1}}, \ldots, S_{v_{i_1}^{k}}$ that $P$ cuts. We claim that the set of vertices $Q=\{v_{i_1}^{1}, \ldots, v_{i_k}^{k}\}$ induce a multi-colored clique in $M$. To show that, it suffices to show that each of the ${k \choose 2}$ edges $rr_e$ cut by $P$, between $r$ and the root of an edge-gadget $T_e$, where $e=v_{i}^{j}v_{p}^{q}$, $i, p \in [n], p, q \in [k], p < q$, satisfies that $v_{i}^{j}, v_{p}^{q} \in Q$.

Consider an arbitrary group $G_j$, $j \in [k]$. The size of $G_j$ is $c_j$, and by \ifshort Lemma~\ref{cor:exactassignment}\fi \iflong Corollary~\ref{cor:exactassignment}\fi, $\f$ assigns the star $S_{v_{i_j}^{j}}$ of size $c_j - (k-1)b_j - (k-1)x_{i_j}$ to $G_j$. Each star $S_e$ is assigned to some group $G_{p}^{q}$ whose size is exactly $|S_e|$. Therefore, each group $G_j$ contains a vertex-gadget and some of the stars $S'_{v_{i'}^{j'}}$, $i' \in [n], j' \in [k]$. Observe that group $G_j$, $j < k$, cannot contain a star $S'_{v_{i'}^{j'}}$ such that $j' > j$ because the size of such a star is at least $b_{j'} > k'^2 b_{j}$, and hence the size of such a star plus the size of $S_{v_{i_j}^{j}}$ would exceed the size of $G_j$. Since there are exactly $k-1$ stars, of the form $S'_{v_{*}^{k}}$ contained in edge-gadgets corresponding to edges incident to class $C_k$, it follows that all these stars must be assigned by $\f$ to group $G_k$. Moreover, no other star $S'_{v_{*}^{j}}$, $j < k$, can be assigned to $G_k$, as the size of such a star would be at least $b_1 > (k-1)x_i$ for any $i \in [n]$; hence, $G_k$ would contain vertex gadget $S_{{v_{i_k}^{k}}}$ of size $c_k - (k-1)b_k - (k-1)x_{i_k}$, plus $k-1$ stars $S'_{v_{*}^{k}}$ of total size greater than $(k-1)b_k$, plus a star of size at least $b_1 > (k-1)x_{i_k}$, and the size of $G_k$ would exceed $c_k$.

Similarly, all the $k-1$ stars of the form $S'_{v_{*}^{k-1}}$ contained in edge-gadgets corresponding to edges incident to class $C_{k-1}$ are assigned to group $G_{k-1}$, and following this argument, we obtain that for each $j \in [k]$, the $(k-1)$ stars of the form $S'_{v_{*}^{j}}$ must be assigned to group $G_j$. We claim that all these stars must correspond to the same vertex $v_{i_j}^{j}$. Observe that this will prove that $Q$ is a clique, since it will imply that each vertex in $Q$ is incident to exactly $k-1$ of the ${k \choose 2}$ many edges between the color classes.

Let $S'_{v_{i'_1}^{j}}, \ldots, S'_{v_{i'_{k-1}}^{j}}$ be the $k-1$ stars placed in $G_j$. The sizes of these stars are $b_j + x_{i'_1}, \ldots, b_j + x_{i'_{k-1}}$, respectively. The size $c_j$ of $G_j$ is equal to the sum of the sizes of these $k-1$ stars, plus that of $S_{v_{i_j}^{j}}$. Therefore: $c_j=c_j - (k-1)b - (k-1)x_{i_j} + (k-1)b +  x_{i'_1} + \cdots + x_{i'_{k-1}}$, and hence, $(k-1)\cdot x_{i_j} = x_{i'_1} + \cdots + x_{i'_{k-1}}$. Since the set $X$ is $(k-1)$-non-averaging, it follows that $x_{i_j} = x_{i'_1} = \cdots = x_{i'_{k-1}}$, and hence, the $(k-1)$ stars $S'_{v_{*}^{j}}$ must correspond to vertex $v_{i_j}^{j}$. \qed
\end{proof}

\ifshort
\begin{theorem}   \label{thm:mainwhard}
\gtp~and~\btp~are $W[1]$-complete.
\end{theorem}
\fi

\iflong
\begin{theorem} \label{thm:mainwhard}
\gtp~is $W[1]$-complete.
\end{theorem}

\begin{proof}
The $W[1]$-hardness result follows from Lemma~\ref{lem:gtptobtp}, Lemma~\ref{lem:togtp} and Lemma~\ref{lem:converse}. To prove membership in $W[1]$, we use the characterization of the class $W[1]$ given by Chen \etal~\cite{yijia}:

\begin{quote}
A parameterized problem $Q$ is in $W[1]$ if and only if there is a computable function $h$ and a nondeterministic \fpt algorithm $\mathbb{P}$ for a nondeterministic-RAM machine deciding $Q$, such that, for each instance $(x, k)$ of $Q$ ($k$ is the parameter), all nondeterministic steps of $\mathbb{P}$ take place during the last $h(k)$ steps of the computation.
\end{quote}

Therefore, to show that \gtp~is in $W[1]$, it suffices to exhibit such a nondeterministic \fpt algorithm $\mathbb{P}$.

Given an instance $I=(T, k, b, s_1, \ldots, s_b)$ of \gtp, where $T$ is assumed to be rooted at an arbitrary vertex $r \in T$, the algorithm $\mathbb{P}$ starts by performing a pre-processing phase. This phase consists of performing a depth-first search on $T$ to compute (and store) descendancy information that allows us to answer, for any two vertices $u, v \in T$, whether or not $u$ is a descendant of $v$ in $T$. (For instance, for each vertex $w \in V(T)$, we can compute a pair of time stamps $(d(w), f(w))$, where $d(w)$ is the discovery time of $w$, and $f(w)$ is the finishing time for $w$, during the depth first search process. It is well known that, for any two vertices $u, v \in V(T)$, $u$ is a descendant of $v$ in $T$ if and only if $d(v) < d(u) < f(u) < f(v)$; for instance, see~\cite{clrs01}.) Moreover, during this pre-processing phase, we compute (and store), for each vertex $v$, the number of vertices in the subtree $T_v$ of $T$ rooted at $v$.

After the above pre-processing phase is complete, $\mathbb{P}$ (nondeterministically) guesses a set $E'$ of $k$ edges $e_1, \ldots, e_k$ from $T$ to be cut; let $e_i=v_iu_i$, for $i \in [k]$, where $v_i$ is the parent of $u_i$. Next, $\mathbb{P}$ determines the number of vertices in each of the $k+1$ components in $T-E'$ as follows. Using the descendancy information computed in the pre-processing phase, and noting that the descendancy relation is a partial order relation on the vertices of $T$, $\mathbb{P}$ constructs a Hasse diagram for this relation (excluding transitive relationships in the representation) whose vertices are $r$, plus the $k$ vertices $u_i$, for $i \in [k]$ (\ie, a Hasse diagram for the descendancy relation restricted to these vertices). The size of the components in $T-E'$ can now be computed by going over the Hasse diagram bottom-up, and for each vertex $u$ in the Hasse diagram, computing the size of the component containing $u$ by subtracting from the number of vertices in $T_u$ (computed and stored during the pre-processing phase) the number of vertices in each of the subtrees of $T$ rooted at the children of $u$ in the Hasse diagram. Finally, after computing the size of each component in $T-E'$, $\mathbb{P}$ tries each of the \fpt-many possible assignments of these components to the groups, or nondeterministically guesses such an assignment (note that the number of groups is at most $k+1$), and accepts if and only if one of these assignments results in groups of sizes $s_1, \ldots, s_b$. Clearly, all the computation done by $\mathbb{P}$ after the pre-processing phase, including the nondeterministic steps, are upper bounded by $h(k)$, where $h$ is a computable function, and hence, $\mathbb{P}$ meets the required conditions in the characterization of $W[1]$ stated above.
\qed
\end{proof}

Next, we show that \gtp~and~\btp~are equivalent modulo \fpt-reducibility.
It is clear that \btp~$\preceq_{fpt}$~\gtp~via an \fpt-reduction that maps an instance $(T, k, b)$ of \btp~to the instance $(T, k, b, s_1=s_2=\cdots=s_{b}=|V(T)|/b)$. The following lemma proves the converse:

\begin{lemma}\label{lem:gtptobtp}
\gtp~$\preceq_{fpt}$ \btp.
\end{lemma}

\begin{proof}
Let $I=(T, k, b, s_1, \ldots, s_{b})$ be an instance of \gtp, let $n=|T|$, and note that $s_1+ \cdots +s_{b}=n$ (otherwise the instance is a no-instance, and we can map it in constant time to a trivial no-instance of \btp). Note also that since removing at most $k$ edges from a tree results in at most $k+1$ components, we can assume, without loss of generality, that $b \leq k+1$.

For each $i \in [b]$, we define $x_i=5n-s_i$, and we create a star $S_i$ with root $r_i$ and $x_i-1$ leaves.  Let $T'$ be the tree obtained from $T$ by rooting $T$ at any vertex $r$, and adding an edge between each root $r_i$ of a star $S_i$ and $r$, for $i \in [b]$. We map the instance $I$ of \gtp~to the instance $I'=(T', k+b-1, b)$. Since $b \leq k+1$, this reduction is clearly an \fpt-reduction. Next, we prove its correctness.

One direction is easy: suppose that $I$ is a yes-instance of \gtp~and we show that $I'$ is a yes-instance of \btp. Since $I$ is a yes-instance of \gtp, there is a solution $P=(E_P, \f)$ to $I$, where $\f$ assigns the components in $T-E_P$, where $|E_P| \leq k$, to $b$ groups $G_1, \ldots, G_{b}$, such that the size of $G_i$ is $s_i$. Let $G_j$, where $j \in [b]$, be the group that contains the root $r$ of $T$. We define the solution $P'=(E'_P, \f')$ to $I'$ as follows. $E'_P$ consists of the set of edges $E_P$ plus each of the $b-1$ edges $rr_i$, where $i \neq j$. The assignment $\f'$ is defined as follows. The assignment agrees with $\f$ on mapping all components, with the difference that the component that contains $r$, though placed in the same group by $\f'$ as by $\f$, now contains the additional subtree $S_j+r_jr$ of $T'$. For the other stars, $\f'$ maps each star $S_i$, $i \neq j$, to group $G_i$. Since group $G_i$ has size $s_i$ in $I$, group $G_i$ in $I'$ has size $s_i+ |S_i| = s_i + x_i =5n$, for $i \in [b]$.

To prove the converse, suppose that $I'$ is a yes-instance of \btp, and let $P'=(E_P', \f')$ be a solution to $I'$. Note that, by constructions, $|T'| =5b\cdot n$, and hence each of the $b$ groups is assigned exactly $5n$ vertices by $\f'$. We first prove that $P'$ cuts at least $b-1$ of the edges $rr_i$, $i \in [b]$. Suppose not, then two substars $S_i$ and $S_j$ remain connected in $T'-E'_P$, and since
the number of edges cut by $P'$ is at most $k+b-1 \leq n + n+1 -1=2n$, the total size of the two substars is at least $x_i + x_j -2n=10n - s_i -s_j -2n > 5n$ (since each of $s_j, s_j \leq n$). Since the two substars will be placed in the same group whose size is exactly $5n$, this is a contradiction. It follows from above that at least $b-1$ of the edges $rr_i$ are cut by $P'$, and by the same arguments made above, each substar resulting from an $S_i$, $i \in [b]$, after removing the edges of $E'_P$ (\ie, $S_i-E'_P$) must be placed in a distinct group from any other substar $S_j$, $i \neq j$. Let $E_P = E'_P \setminus \{rr_i \mid i \in [b]\}$. It follows from above that $|E_P| \leq b + k-1 - (b-1) \leq k$.

Next, we prove that we can assume that $P'$ cuts no star-edge from any $S_i$, for $i \in [b]$. Note that proving the aforementioned statement completes the proof since it will imply that each group $G_i$ containing $S_i$ must contain components in $T-E_P$ of total size $5n-x_i=s_i$, and those components must constitute all the components of $T-E_P$; this will show that $I$ has a solution $E_P$ that cuts at most $k$ edges. Suppose that a group $G_i$, for $i \in [b]$, contains a (proper) component of $S_j$, for some $j \neq i$. Since a substar of $S_j$ is in $G_j$, it follows that the proper part of $S_j$ contained in $G_i$ consists of leaves from $S_j$, whose edges were cut by $P'$. We claim that there must exist a group $G_q$, $q \in [b]$, containing a substar $S_q$ such that $P'$ cuts at least one edge from star $S_q$, and such that $G_q$ contains a component of $T-E_P$. If this is not the case, then each component of $T-E_P$ must appear with a complete star $S_i$, $i\in [b]$. Since the sum of the sizes of all components of $T-E_P$ is $|T|=n$, each $S_i$ has size $5n-s_i$, and no two substars appear in the same group, it follows that no edge from any $S_i$ is cut, which is a contradiction to our assumption that $P'$ cuts star-edges.

Now let $G_q$ be a group containing a substar $S_q$ and a component $C$ from $T-E_P$. We modify $P'$ as follows. Since $T$ is a tree, $C$ must contain a leaf $u$; let $\pi(u)$ be the parent (if $|C| \neq 1$) of $u$ in $C$. We modify $E'_P$ by removing an edge, say $wr_q$ that it cuts from $S_q$, and adding to $E'_P$ the edge $u\pi(u)$ if $\pi(u)$ exists, thus cutting the edge $u\pi(u)$ in $P'$. We then modify $\f'$ by placing $u$ in the group that contained $w$ (note that $w$ is now attached to $S_q$ in the modified solution). This results in anther solution of $I'$ that cuts no more than the number of edges in $E_P$. Repeating this argument, we end up with a solution to $I'$ that does not cut any edge from any star $S_i$, for $i \in [b]$. The above argument shows that it can be assumed that $P'$ cuts no star-edge from any $S_i$, for $i \in [b]$, and completes the proof.
\qed
\end{proof}
\fi


\iflong
\begin{corollary}
\label{cor:btpwcomplete}
\btp~is $W[1]$-complete.
\end{corollary}

\begin{proof}
The $W[1]$-hardness follows from the $W[1]$-hardness of \gtp~proved in Theorem~\ref{thm:mainwhard} and Lemma~\ref{lem:gtptobtp}. Membership in $W[1]$ follows from that of \gtp, proved in Theorem~\ref{thm:mainwhard}, and the fact that \btp~is a restriction of~\gtp, as observed above. \qed
\end{proof}
\fi
\iflong
\section{Subexponential-time Algorithms for \gtp~and~\btp}
\label{sec:algorithms}
\fi
\ifshort
\section{Subexponential-time Algorithms}
\label{sec:algorithms}
\fi
Let $n \in \integ$. A \emph{partition} of $n$ is a collection $X$ of positive integers such that $\sum_{x \in X} x =n$. Let $p(n)$ denote the total number of (distinct) partitions of $n$. It is well known that $p(n) = 2^{\Oh(\sqrt{n})}$~\cite{hardy}. It follows that the total number of partitions of all integers $n'$, where $0 < n' \leq n$, is $\sum_{0 < n' \leq n} p(n') =2^{\Oh(\sqrt{n})}$.

Let $L$ be a list of numbers in $\nat$ that are not necessarily distinct. \iflong (Note that a list may contain zeros.)\fi We denote by $L(i)$ the $i$th number in $L$, and by $L_i$ the sublist of $L$ consisting of the first $i$ numbers. The \emph{length} of $L$, denoted $|L|$, is the number of elements in $L$.

Let $(T, k, b, s_1, \ldots, s_b)$ be an instance of \gtp. Let $n=|T|$. 
Consider a partial assignment of $n' \leq n$ vertices of $T$ to the $b$ groups, with the possibility of some groups being empty. Since the groups are indistinguishable, such an assignment corresponds to a partition of the $n'$ vertices into at most $b$ parts, and can be represented by a \emph{sorted} list $L$ of $b$ numbers in $\nat$ whose sum is $n'$, where $L(i) \leq n'$ for $i \in [b]$, is the number of vertices assigned to group $i$; we call such a representation of the groups, under a partial assignment, a \emph{size representation}, denoted as \srep. Note that the zeroes in a \srep~appear at the beginning. Since each \srep~corresponds uniquely to a partition of a number $n' \leq n$ prefixed by less than $b \leq n$ zeroes, it follows that the total number of $\sigma$-representations is $n \cdot 2^{\Oh(\sqrt{n})}= 2^{\Oh(\sqrt{n})}$.

Let $X, Y, Z$ be three lists of the same length. We write $X= Y \pp Z$ if there is a list $Y'$ obtained via a permutation of the numbers in $Y$, and a list $Z'$ obtained via a permutation of the numbers in $Z$, such that $X(i) = Y'(i) + Z'(i)$, for every $i \in [|X|]$; that is, in the context when the lists are \sreps, $X= Y \pp Z$ if each group-size in $X$ can be obtained, in a one-to-one fashion, by adding a group-size in $Y$ to a group-size in $Z$ (including group-sizes zero). \ifshort We have:\fi

\ifshort
\begin{proposition}
 \label{prop:3vectors}
 There is a subroutine \CR($X, Y, Z$) that determines if $X = Y \pp Z$ in time $2^{\Oh(\sqrt{n})}$.
\end{proposition}
\fi

\iflong
Let $n \in \nat$, and let $X, Y, Z$ be three \sreps.  We wish to decide if $X = Y \pp Z$. To do so, we apply the subroutine \CR($X, Y, Z$) described in the next section.

\subsection{The Subroutine \CR} \label{subsec:cr}
The subroutine \CR($X, Y, Z$), for \sreps~$X, Y, Z$, uses dynamic programming. It constructs a table $\Lambda$, where for each $ i \in [b]$, and each pair of \sreps~$V$, $W$, we have an entry $\Lambda(X_i, V_i, W_i)$, where $\Lambda(X_i, V_i, W_i)$ is \true~if and only if $X_i=V_i \pp W_i$; if $\Lambda(X_i, V_i, W_i)$ is \true, we also store a witness to the decomposition of $X_i$ into $V_i$ and $W_i$ that we compute during the dynamic programming process. To construct $\Lambda$, we iterate over all values $i= 1, \ldots, b$. For each value $i$, we iterate through all pairs of \sreps~$V, W$. It is clear how to populate the table for $i=1$, and each pair of \sreps~$V, W$, as in this case $\Lambda(X_1, V_1, W_1)$ is true if and only if $X_1(1) = V_1(1) + W_1(1)$. Suppose, inductively, that we have populated the table $\Lambda$, for every $1 \leq j < i$, and every pair of \sreps~$V$ and $W$. To populate the entry $\Lambda(X_i, V_i, W_i)$, for a fixed pair of \sreps~$V$ and $W$, enumerated from among all possible pairs of \sreps, we do the following. We iterate through every pair of integers $j, k$ satisfying $1 \leq j, k \leq i$. If for a pair $j, k$ we have (1) $X_i(i) = V_i(j) + W_i(k)$ and (2) $\Lambda(X'_{i-1}, V'_{i-1}, W'_{i-1})$ is \true, where $X'_{i-1}$ is the list obtained by removing $X_i(i)$ from $X_i$, and  $V'_{i-1}$ and $W'_{i-1}$ are the lists obtained by removing $V_i(j)$ and $W_i(k)$ from $V_i$ and $W_i$ (without changing the respective sorted order of the remaining elements in $V_i$ and $W_i$), respectively, then we set $\Lambda(X_i, V_i, W_i)$ to \true; otherwise, if no such pair of numbers $j, k$ exists, we set $\Lambda(X_i, V_i, W_i)$ to \false.

\begin{proposition}
\label{prop:3vectors}
Let $X, Y, Z$ be three \sreps. The subroutine \CR($X, Y, Z$) determines if $X = Y \pp Z$ in  time $2^{\Oh(\sqrt{n})}$.
\end{proposition}

\begin{proof}
The correctness of \CR($X, Y, Z$) follows by a straightforward induction on $|X|$. To analyze its running time, observe that in  \CR($X, Y, Z$), we iterate $b \leq n$ times. In iteration $i$, $i \in [b]$, we enumerate every pair of \sreps~$V, W$, and for each pair $V, W$, we iterate through every pair of integers $j, k$ satisfying $1 \leq j, k \leq i$. The total number of pairs of \sreps~is $2^{\Oh(\sqrt{n})} \cdot 2^{\Oh(\sqrt{n})}=2^{\Oh(\sqrt{n})}$, and the total number of pairs of integers $j, k$ satisfying $1 \leq j, k \leq i$ is $\Oh(n^2)$. For a fixed pair of \sreps~$V, W$, and a fixed pair of integers $i, j$, all the operations performed in the dynamic programming process can be implemented in $n^{O(1)}$ time.
The running time of \CR($X, Y, Z$) is thus upper bounded by $2^{\Oh(\sqrt{n})} \cdot n^{O(1)}= 2^{\Oh(\sqrt{n})}$.  \qed
\end{proof}
\fi
\iflong
\subsection{The Exact Algorithm}\label{subsec:algo}
 In this section, we present a subexponential-time algorithm for \gtp, and hence for \btp, that runs in $2^{\Oh(\sqrt{n})}$ time, where $n =|V(T)|$.\fi~Let $(T, k, b, s_1, \ldots, s_b)$ be an instance of \gtp. The key observation that leads to a subexponential-time algorithm is that the $b$ groups are indistinguishable. Therefore, all assignments of the $n$ vertices in $T$ to the $b$ groups can be compactly represented by lists of numbers, where each list corresponds to a partition of $n$ into $b$ parts. This simple, yet crucial, observation allows for a ``compact representation'' of all solutions using a solution space of size $2^{\Oh(\sqrt{n})}$. \iflong Intuitively speaking, this solution space consists of solutions corresponding to all possible partitions of $n$, whose number is $p(n)$. We start by giving an intuitive description of the algorithm, then proceed to the details.\fi

Suppose that $T$ is rooted at an arbitrary vertex $r$. The algorithm uses dynamic programming, starting from the leaves of $T$, and climbing $T$ up to its root $r$.
At each vertex $v$ in $T$, we construct a table $\Gamma_v$ that contains the following information. For each \srep~$X$, for each $k' =0, \ldots, n$, and for each $s \in [n]$,
$\Gamma_v(k', X, s)$ is \true~if and only if there is a cut $C$ of $k'$ edges in $T_v$ \iflong (the subtree of $T$ rooted at $v$) \fi such that the component $P_v$ containing $v$ in $T_v -C$ has size $s$ \ifshort (note that this component, so far, is still attached to the rest of the tree above $v$)\fi, and such that there is an assignment to the components in $T_v - C - P_v$ to the $b$ groups whose \srep~is $X$; otherwise, $\Gamma_v(k', X, s)$ is \false. If $\Gamma_v(k', X, s)$ is \true, we store a witness that realizes such a partial solution\iflong (\ie, we store a set $C$ of $k'$ edges in $T_v$ such that the component $P_v$ containing $v$ in $T_v -C$ has size $s$, and we store an assignment to the components in $T_v - C-P_v$ to the $b$ groups whose \srep~is $X$)\fi.  To compute $\Gamma_v$, we consider the children of $v$ one by one. After a child $u_i$ of $v$ is considered, we have computed a partial table $\Gamma_i$ containing partial solutions up to child $u_i$; this is done by considering the two possibilities of whether or not the edge $vu_i$ is in the cut $C$. Although the above may seem like we are enumerating all possibilities for the edges between $v$ and its children to be cut or not, the crucial ingredient for this approach to achieve the desired running time is that the table $\Gamma_v$---at vertex $v$---can be computed based on the tables corresponding to the children of $v$ in $2^{\Oh(\sqrt{n})}$ time. This analysis works similarly to iterative compression, as the table $\Gamma_i$ \iflong, computed after child $u_i$ has been considered,\fi is a compressed table, storing $2^{\Oh(\sqrt{n})}$ many entries\iflong, regardless of the status of the edges between $v$ and its children considered so far. We proceed to the details\fi.

\iflong Recall that, for each vertex $v \in T$, for each $k'=0, \ldots, n$, for each $s \in [n]$, and for each \srep~$X$, $\Gamma_v(k', X, s)$ is \true~if and only if there is a cut $C$ of $k'$ edges in $T_v$, and an assignment to the components in $T_v -C -P_v$, with $s$ being the size of the component $P_v$ containing $v$ in $T_v - C$ (note that this component, so far, is still attached to the rest of the tree above $v$), that realizes the \srep~$X$. The dynamic programming algorithm proceeds in a bottom-up fashion, from the leaves of $T$ to its root $r$. \fi Suppose that the algorithm is at vertex $v$ whose children are $u_1,\ldots, u_d$, and that the tables $\Gamma_{u_1}, \ldots, \Gamma_{u_d}$ associated with $u_1,\ldots, u_d$, respectively, have been constructed. To compute $\Gamma_v$, we iterate through the edges $vu_1, \ldots, vu_d$. Let $T_p$, for $p=1, \ldots, d$, be the subtree of $T$ rooted at $v$ that is induced by the vertex-set $(\bigcup_{j=1}^{p} V(T_{u_j})) \cup \{v\}$. Consider edge $vu_i$, and assume inductively, that a table $\Gamma_{i-1}$ has been computed (based on tables $\Gamma_{u_1}, \ldots, \Gamma_{u_{i-1}}$) that contains the following information. For each $k'=0, \ldots, n$, for each $s \in [n]$, and for each \srep~$X$, $\Gamma_{i-1}(k', X, s)$ is \true~if and only if there is a cut $C$ of $k'$ edges in $T_{i-1}$, with $s$ being the size of the component $P_v$ containing $v$ in $T_{i-1} - C$, and an assignment to the components in $T_{i-1} -C -P_v$ that realizes $X$\iflong; if $\Gamma_{i-1}(k', X, s)$ is \true, we also store a witness to the partial solution\fi. After considering $vu_i$, we will compute a table $\Gamma_i$ such that, for each $k'=0, \ldots, n$, for each $s \in [n]$, and for each \srep~$X$, $\Gamma_{i}(k', X, s)$ is \true~if and only if there is a cut $C$ of $k'$ edges in $T_{i}$, with $s$ being the size of the component $P_v$ containing $v$ in $T_{i} - C$, and an assignment to the components in $T_{i} -C-P_v$ that realizes $X$\iflong; if $\Gamma_{i}(k', X, s)$ is \true, we also store a witness to the partial solution\fi. We explain how the Boolean value $\Gamma_{i}(k', X, s)$ is computed, and omit how the witness \ifshort is stored\fi \iflong can be stored, as this is straightforward\fi. After we are done computing $\Gamma_{d}$, we set $\Gamma_v = \Gamma_{d}$.

To compute $\Gamma_i$, we compute two tables $\Gamma_{i}^{-}$ and $\Gamma_{i}^{+}$, and set $\Gamma_i= \Gamma_{i}^{-} \cup \Gamma_{i}^{-}$. Table $\Gamma_{i}^{-}$ contains the solutions that can be obtained by cutting edge $uv_i$, and $\Gamma_{i}^{+}$ contains those that can be obtained by not cutting edge $vu_i$. \iflong We explain next how each of $\Gamma_{i}^{-}$ and $\Gamma_{i}^{+}$ is computed.\fi

\noindent {\bf 1.}\  To compute $\Gamma_{i}^{-}$, we enumerate each possible triplet $(k', X, s)$, where $k' =0, \ldots, n$, $s \in [n]$, and $X$ is a \srep.  Fix such a triplet $(k', X, s)$. To compute $\Gamma_{i}^{-}(k', X, s)$, we iterate through every entry in $\Gamma_{u_i}$ containing $(k_{u_i}, Y, s_{u_i})$ and every entry of $\Gamma_{i-1}$ containing $(k_{i-1}, Z, s_{i-1})$ such that $k'=k_{u_i}+k_{i-1} +1$ (because 1 more cut is introduced, corresponding to the edge $vu_i$), and $s = s_{i-1}$ because the component $P_{u_i}$ containing $u_i$ of size $s_{u_i}$ becomes a separate component after $vu_i$ is cut. Since $P_{u_i}$ becomes a separate component, it will be placed into one of the groups, and hence, it contributes its size to one of the numbers in the \srep~$X$. We enumerate each number in $X$ as the number that $P_{u_i}$ contributes to. For each number $j$ in $X$ satisfying $j \geq |P_{u_i}|$, we subtract $|P_{u_i}|$ from $j$ in $X$ to obtain a new \srep~$X'$ from $X$, and then call \CR($X', Y, Z$); $\Gamma_{i}^{-}(k', X, s)$ is \true~iff for some number $j$ in $X$, \CR($X', Y, Z$) returns \true.

\noindent {\bf 2.} To compute $\Gamma_{i}^{+}$, we enumerate each triplet $(k', X, s)$, where $k' =0, \ldots, n$, $s \in [n]$, and $X$ is a \srep. Fix such a triplet $(k', X, s)$. To compute $\Gamma_{i}^{+}(k', X, s)$, we iterate through every entry in $\Gamma_{u_i}$ containing $(k_{u_i}, Y, s_{u_i})$, and every entry in $\Gamma_{i-1}$ containing $(k_{i-1}, Z, s_{i-1})$, such that $k'=k_{u_i}+k_{i-1}$, and $s = s_{u_i}+ s_{i-1}$ (because $s_{u_i}$ is attached to $v$). We call \CR$(X, Y, Z$), and set $\Gamma_{i}^{+}(k', X, s)$ to \true~iff \CR($X, Y, Z$) returns \true. \\

\iflong
When the table $\Gamma_r$, at the root $r$ of $T$, has been computed, we iterate through the entries in $\Gamma_r$ to determine if for the desired value of $k$ (or the minimum value of $k$ in case we are interested in solving the optimization version of the problem), an entry $\Gamma_r(k, X, s)$ is \true, such that one of the sorted lists $Y_1, \ldots, Y_b$, is $[s_1, \ldots, s_b]$, where $Y_i$, $i \in [b]$, is the list obtained by adding $s$ to $X(i)$ and sorting the resulting list (assuming, without loss of generality that $s_1 \leq s_2 \cdots \leq s_b$), and return a witness to the solution (if the solution exists); otherwise, we return \false. Note that \btp~is the restriction to \gtp~to instances in which $s_1 = \cdots = s_b = |V(T)|/b$, and hence can be solved by the same algorithm.
\fi
\ifshort \vspace*{-4mm}\fi
\begin{theorem} \label{thm:exactalgorithm}
The dynamic programming algorithm described above solves \gtp~and \btp~in time $2^{\Oh(\sqrt{n})}$.
\end{theorem}

\iflong
\begin{proof}
The correctness of the algorithm follows by an inductive proof showing that the invariant properties about the tables $\Gamma_i$, obtained during the computation at vertex $v \in T$, hold true, assuming that they hold true for table $\Gamma_{i-1}$ and the tables at the children of $v$.

To analyze the running time of the algorithm, it suffices to show that the computation of $\Gamma_v$, at any vertex $v \in T$, takes $2^{\Oh(\sqrt{n})}$ time, as then the overall running time of the algorithm is upper bounded by $n^{O(1)} \cdot 2^{\Oh(\sqrt{n})}=2^{\Oh(\sqrt{n})}$. To compute
$\Gamma_v$, at a vertex $v \in T$ with $d$ children $u_1, \ldots, u_d$, the algorithm iterates through each child $u_i$ of $v$, and computes the table $\Gamma_i$ from the two tables $\Gamma_{i-1}$ and $\Gamma_{u_i}$, the size of each is $2^{\Oh(\sqrt{n})}$. To compute $\Gamma_i$, the algorithm computes two tables $\Gamma_{i}^{-}$ and $\Gamma_{i}^{+}$, by distinguishing whether or not edge $vu_i$ is cut or not, and takes their union. The computation of each of these two tables is done by enumerating all triplets $(k', X, s)$, where $k'=0, \ldots, n$, $X$ is a \srep, and $s \in [n]$. Since the number of \sreps~is $2^{\Oh(\sqrt{n})}$, the total number of these triplets is $O(n^2) \cdot 2^{\Oh(\sqrt{n})}=2^{\Oh(\sqrt{n})}$. For each triplet $(k', X, s)$, to compute $\Gamma_{i}^{-}(k', X, s)$ (resp.~$\Gamma_{i}^{+}(k', X, s)$), the algorithm enumerates all entries in each of $\Gamma_{i-1}$ and $\Gamma_{u_i}$; there are $2^{\Oh(\sqrt{n})}\cdot 2^{\Oh(\sqrt{n})}=2^{\Oh(\sqrt{n})}$ many entries. The algorithm then performs some polynomial-time computation and calls \CR(), which runs in time $2^{\Oh(\sqrt{n})}$. Therefore, computing each of the two tables $\Gamma_{i}^{-}$ and $\Gamma_{i}^{+1}$, and hence the table $\Gamma_{i}$, takes time $n^{O(1)} \cdot 2^{\Oh(\sqrt{n})}=2^{\Oh(\sqrt{n})}$. It follows that computing $\Gamma_v$ takes time $d \cdot 2^{\Oh(\sqrt{n})} = 2^{\Oh(\sqrt{n})}$. Note that a crucial property to obtain this running time is that the size of each table $\Gamma_v$, for $v \in T$, remains $2^{\Oh(\sqrt{n})}$, because the number of \sreps, and hence triplets $(k, X, s)$, is $2^{\Oh(\sqrt{n})}$.\qed
\end{proof}
\fi

\bibliographystyle{plain}
\bibliography{ref}

\end{document}